\documentclass{lmcs}
\pdfoutput=1

\usepackage{lastpage}
\lmcsdoi{15}{4}{16}
\lmcsheading{}{\pageref{LastPage}}{}{}%
{May~11,~2017}{Dec.~19,~2019}{}

\usepackage{hyperref}
\usepackage{amssymb}
\usepackage{stmaryrd}
\usepackage{xspace}
\usepackage{booktabs}

\DeclareMathAlphabet{\mathpzc}{OT1}{pzc}{m}{it}

\newcommand{\set}[1]{\left\{#1\right\}}
\newcommand{\tuple}[1]{\left(#1\right)}

\newcommand{\MSOT}{\text{MSOT}\xspace}
\newcommand{\MSOTs}{\text{MSOTs}\xspace}
\newcommand{\FOT}{\text{FOT}\xspace}
\newcommand{\MSO}{\text{MSO}\xspace}
\newcommand{\MSOeq}{\text{MSO}$[=]$\xspace}
\newcommand{\FO}{\text{FO}$[<]$\xspace}
\newcommand{\MSOc}{\text{MSO}$_\con$\xspace}
\newcommand{\FOc}{\text{FO}${[<]}_\con$\xspace}
\newcommand{\con}{\mathsf{c}}
\newcommand{\F}{$\mathcal{F}$\xspace}
\newcommand{\Fc}{$\mathcal{F}_\con$\xspace}

\newcommand{\FOd}{\text{FO}$^2{[<]}$\xspace}
\newcommand{\FOdc}{\text{FO}$^2{[<]}_\con$\xspace}

\newcommand{\Bsig}[1]{$\mathcal B \Sigma _{#1}[<]$\xspace}
\newcommand{\Sig}[1]{$\Sigma _{#1}[<]$\xspace}
\newcommand{\Bsigc}[1]{$\mathcal B \Sigma _{#1}{[<]}_\con$\xspace}
\newcommand{\FOsuc}{\text{FO}$[+1]$\xspace}

\newcommand{\dom}{\text{dom}} 
\newcommand\sem[1]{\llbracket#1\rrbracket}

\newcommand{\pspace}{\textsc{PSpace}\xspace}
\newcommand{\ptime}{\textsc{PTime}\xspace}
\newcommand{\exptime}{\textsc{ExpTime}\xspace}

\newcommand\aut{\mathcal{A}}
\newcommand\laut{\mathcal{L}}
\newcommand\raut{\mathcal{R}}

\newcommand\trans{\mathcal{T}}
\newcommand\bim{\mathcal{B}}

\newcommand\var{\textbf{C}\xspace}
\newcommand\ap{\mathbf{A}}
\newcommand\id{\mathbf{I}}
\newcommand\fin{\mathbf{F}}
\newcommand\da{\mathbf{DA}}
\newcommand\jtrivial{\mathbf{J}}

\newcommand\out{\mathpzc{o}}
\newcommand\bout{\omega}
\newcommand\init{\mathpzc{i}}
\newcommand\final{\mathpzc{t}}
\newcommand\lfinal{\lambda}
\newcommand\rfinal{\rho}

\newcommand\dist[1]{\parallel\!\! {#1}\!\!\parallel}
\newcommand\cla[1]{[#1]}
\newcommand{\leftcan}{\mathit{Left}}
\newcommand{\rightcan}{\mathit{Right}}
\newcommand{\finer}{\sqsubseteq}

\newcommand{\bimeq}{\mathrel{\boxminus}}

\newcommand{\leftcong}{\mathrel{\leftharpoonup}}
\newcommand{\rightcong}{\mathrel{\rightharpoonup}}

\newcommand{\contribution}[1]{\underline{#1}}

\newcommand{\struct}{\mathcal{S}}
\newcommand{\lab}{label}

\newcommand{\good}{$\epsilon$-isolating\xspace}

\newcommand{\ie}{\emph{i.e.}\xspace}
\newcommand{\eg}{\emph{e.g.}\xspace}

\usepackage{graphicx}
\usepackage{subfigure}

\usepackage{xcolor,pgf,tikz,pgffor}
\usetikzlibrary{automata,decorations,arrows,patterns}

\begin{document}

\title[rational transductions]{Logical and algebraic characterizations\texorpdfstring{\\}{} of rational transductions}

\author[E. Filiot]{Emmanuel Filiot\rsuper{a}}
\address{\lsuper{a}Universit\'e Libre de Bruxelles}
\email{efiliot@ulb.ac.be}
\thanks{This work was partially supported by the ANR \emph{ExStream}
    (ANR-13-JS02-0010) and DeLTA (ANR-16-CE40-0007) projects, 
    the ARC project \emph{Transform}
    (French speaking community of Belgium) and the Belgian FNRS CDR project
    \emph{Flare}. Emmanuel Filiot is research associate at
    F.R.S.-FNRS.}

\author[O. Gauwin]{Olivier Gauwin\rsuper{b}}
\address{\lsuper{b}Universit\'e de Bordeaux, LaBRI, CNRS}
\email{\{olivier.gauwin,nlhote\}@labri.fr}

\author[N. Lhote]{Nathan Lhote\rsuper{{a,b}}}

\subjclass{Theory of computation\~ Transducers} 

\begin{abstract}
  Rational word languages can be defined by several equivalent means:
  finite state automata, rational expressions, finite congruences, or
  monadic second-order (MSO) logic.
  The robust subclass of aperiodic languages is defined by:
  counter-free automata, star-free expressions, aperiodic (finite) congruences,
  or first-order (FO) logic.
  In particular, their algebraic characterization by aperiodic congruences
  allows to decide whether a regular language is aperiodic.

  We lift this decidability result to rational transductions,
  i.e., word-to-word functions defined by finite state transducers.
  In this context, logical and algebraic characterizations
  have also been proposed.
  Our main result is that one can decide if a rational transduction (given as
  a transducer) is in a given decidable congruence class.
  We also establish a transfer result from logic-algebra equivalences over languages to equivalences over transductions.
  As a consequence, it is decidable if a rational transduction is first-order definable,
  and we show that this problem is PSPACE-complete.

\end{abstract}

\maketitle
\section*{Introduction}

\begin{table}
  \small
  \begin{tabular}{@{}l|l|l|l} 
    & Machine & Logic & Algebra \\
    & & & (Canonical object) \\ \toprule

    \multicolumn{3}{l}{\textit{Languages}} &
    \begin{minipage}{4.5cm}
      syntactic congruence =\\
      minimal automaton~\cite{Myhill57,Nerode58}
    \end{minipage}
    \\ \midrule

    regular & automata & \MSO \hfill\cite{Buchi60} & finite \\
    star-free & aperiodic automata & \FO \hfill\cite{McnaughtonP71} & aperiodic~\cite{Schutzenberger65} \\
    \var & \var-automata & \F & \var \\ \midrule

    \multicolumn{3}{l}{\textit{Sequential transductions}} &
    \begin{minipage}{4.5cm}
      syntactic congruence~\cite{Choffrut03}\\
      = minimal transducer
    \end{minipage}
    \\ \midrule

    sequential & sequential transducers & &
    finite \\
    aperiodic & ap.\ sequential transducers & & aperiodic \\
    \var & \var-sequential transducers & & \var \\
    \midrule

    \multicolumn{3}{l}{\textit{Rational transductions}} &
    \begin{minipage}{4.5cm}
    left/right synt.\ congruence,\\
    \contribution{minimal bimachines}
    \end{minipage}
    \\ \midrule

    rational & transducers & op-\MSOT \hfill\cite{Bojanczyk14} & finite~\cite{ReutenauerS91} \\
    aperiodic & ap.\ transducers & \contribution{op-\FOT} \hfill\cite{Bojanczyk14} & aperiodic \\
    \var & \var-transducers & \contribution{\F} & \var \\
    \midrule

    \multicolumn{3}{l}{\textit{Regular transductions}} & ? \\ \midrule 

    regular & 2-way trans., SST & \MSOT \hfill\cite{EngelfrietH01,AlurC10} & \\
    aperiodic & ap. 2-way trans., ap. SST & \FOT \hfill\cite{CartonD15,FiliotKT14} & \\
    \midrule

  \end{tabular}\\
  \var is any congruence class, and \F its associated logic.
  \emph{op} stands for \emph{order-preserving}.\\
  Our contributions concern the underlined objects.
  \caption{Logic-automata-algebra characterizations, for languages and transductions.\label{table:links}}
\end{table}

\subsection*{Logical and algebraic characterizations of rational languages}
The theory of rational languages over finite words is mature and rich with results that stem from their many different characterizations, such as: finite automata, rational expressions, monadic second-order logic (\MSO) and congruences of finite index, as illustrated in Table~\ref{table:links}.
One of the main features of the algebraic approach is the existence, for any rational language, of a canonical object: the \emph{syntactic congruence} (also known as Myhill-Nerode congruence) which is minimal in a strong sense that it is coarser than any congruence recognizing the language, and is related to the minimal automaton of a language.
Furthermore, many correspondences between some \emph{logical
  fragments} of \MSO and \emph{congruence varieties} (see
e.g.~\cite{Straubing94}) have been established. Congruence varieties
are sets of congruences of finite index with good closure properties
including closure under coarser congruences\footnote{If some
  congruence $\sim$ is in some variety $V$, any congruence coarser than
  $\sim$ is in $V$ as well.} which means that the syntactic congruence carries, in some sense, the intrinsic algebraic properties of a language.
The most prominent of the aforementioned results, which was obtained together by Sch\"utzenberger~\cite{Schutzenberger65} and McNaughton and Papert~\cite{McnaughtonP71}, is the correspondence between languages definable in \emph{first-order logic} (\FO) and languages recognizable by a finite \emph{aperiodic} congruence.
Since then many other such correspondences have been found (again, see~\cite{Straubing94}), for instance between the fragment of first-order logic with only two variables and the congruence variety known as $\da$~\cite{TherienW98}.
The strength of these results is that they give, through the syntactic congruence, an \emph{effective} way to decide if a regular language can be expressed in a given logical fragment. The goal of this article is to lift some of these logic-algebra correspondences from rational languages to rational transductions.

\subsection*{Rational transductions}
Transductions are (partial) functions from finite words to finite words.
At the computational level they are realized by transducers, which are automata given together with an output function mapping transitions to finite words.
Although transducers have been studied for almost as long as automata, far less is known about them. A more recent result of~\cite{EngelfrietH01} has sparked some renewed interest in their logical aspect.
This result states that transductions definable by MSO-transducers (\MSOT), a logical model of transducers defined by Courcelle in the general context of graph transductions~\cite{CourcelleE12}, exactly characterize the transductions realized by two-way transducers.
Another model of one-way transducers with registers, called \emph{streaming string transducers}, has also been shown to capture the same class of transductions called \emph{regular transductions}~\cite{AlurC10}.
The first steps of a logic-algebra relationship for transductions have been made by~\cite{FiliotKT14} and~\cite{CartonD15}, where first-order transducers (\FOT) are shown to be expressively equivalent to streaming string transducers (resp.\ two-way transducers) with an aperiodic transition congruence.
These characterizations are however not effective, and the class of regular functions still lacks a canonical object similar to the syntactic congruence, which would yield a way to decide if a given transducer realizes a transduction definable in first-order logic, for example.

\emph{Rational transductions} are the transductions realized by
one-way transducers, and they also coincide with transductions
definable by rational expressions over the product of two free monoids
(with component-wise concatenation and Kleene star)~\cite{BerstelB79}.
In this paper, we define subclasses of rational transductions by
restricting the class of \emph{transition congruence} of the
transducers which define them. The transition congruence of a
transducer is just defined as the transition congruence of its
underlying automaton, ignoring the output. The reason for choosing such
a definition is because we obtain nice correspondences between
classes of transition congruences of transducers and logical fragments
of MSO-transducers. First of all, the rational transductions are
exactly captured by a natural restriction of \MSOT, called
\emph{order-preserving} \MSOT\cite{Bojanczyk14,Filiot15}.  Its first-order fragment was shown to
coincide with rational transduction definable by a one-way transducer
with aperiodic transition congruence in~\cite{Bojanczyk14}. Other
correspondences, for instance with  \FOd, are shown in this paper.

Let us point out that the characterization of first-order rational
functions as the functions defined by aperiodic transducers~\cite{Bojanczyk14}
is not effective since it does not
provide a way to decide if a given transduction is definable in
first-order logic. In~\cite{ChoffrutG14} a different kind of
characterization of (non-functional) transductions over a unary alphabet is given, in
terms of rational expressions rather than in terms of transition
congruence.

A first result on the existence of a canonical machine for
transductions is given in~\cite{Choffrut03}, where the author defines a
construction of a syntactic congruence for sequential transducers,
\emph{i.e.} transducers with a deterministic underlying (input) automaton, which define a strict subclass of rational transductions.
This syntactic congruence, like for languages, is minimal in a
stronger sense than just yielding a sequential transducer with the
minimal number of states: it is minimal in the algebraic sense that it
is coarser than any transition congruence of a sequential transducer
realizing the transduction.
A second such result is the existence of a canonical machine for rational functions, shown in~\cite{ReutenauerS91}.
However, as the authors mentioned in this article, this canonical machine is not minimal in any of the ways described above.
In the present paper we refine the approach of~\cite{ReutenauerS91} in order to obtain minimal machines in terms of both number of states and of transition congruence.
This approach is based on a computational model of transductions called \emph{bimachines}, which captures the rational functions.
A bimachine is a deterministic model of transductions, introduced in~\cite{Schutzenberger61} and further studied (and named) in~\cite{Eilenberg74}, which can be seen as a sequential transducer with regular look-ahead, where the look-ahead is given by a co-deterministic automaton.

In contrast to considering restrictions of the transition congruence
of the underlying automata of transducers, a different way of defining subclasses of rational transductions,
taking into account the outputs, has been proposed in~\cite{CadilhacKLP15}
and~\cite{CadilhacCP17}. It is based on the natural notion of continuity,
\ie, preserving a class of languages by inverse image. In~\cite{CadilhacKLP15}, the authors were able to effectively
characterize sequential transductions definable by AC$^0$ circuits. In~\cite{CadilhacCP17} the authors show how to decide if a transduction
is continuous with respect to many usual varieties and they also
compare the two notions of continuity and realizability.

\subsection*{Contributions}
We first restate in our context the results from~\cite{Choffrut03} which give, for any sequential transduction, a syntactic congruence and thus a minimal transducer. Hence we show that for any congruence class\footnote{We call congruence class a set of congruences of finite index closed under intersection and coarser congruences.} $\var$, $\var$-sequentiality (\ie, realizability by a sequential transducer with a transition congruence in $\var$) can be reduced to deciding if the syntactic congruence belongs to $\var$.
However the syntactic congruence of a sequential transduction may not
capture all the algebraic properties of the transduction, in the
following way: there exists a congruence class $\var$ and a sequential
transduction which is $\var$-rational (realizable by a transducer with
transition congruence in $\var$, i.e., a $\var$-transducer) yet not \var-sequential. On the other hand we show that it is not the case for $\ap$ (the class of aperiodic congruences): any sequential $\ap$-rational function is $\ap$-sequential.

We refine techniques used in~\cite{ReutenauerS91} to give an algebraic characterization of rational functions.
It is known~\cite{ReutenauerS95}, for a congruence class \var, that \var-bimachines and unambiguous \var-transducers realize the same class of transductions.
It is also known that a functional transducer can be disambiguated, however the known disambiguation algorithms do not preserve the class of the transition congruence.
We nevertheless are able to show that unambiguous \var-transducers are as expressive as functional \var-transducers.

The most important result of this paper is an effective way to compute, for any rational function, a finite set of minimal bimachines. These bimachines are minimal in the strong sense that they have the coarsest transition congruences among all bimachines (and transducers) realizing the function. As a corollary this gives a way to decide if a rational transduction can be realized by a transducer with transition congruence in any given decidable congruence class \var, just by checking if one of the minimal bimachines is a \var-bimachine.

We define a logical formalism to associate with any logical fragment \F of \MSO a class of \F-definable transductions.
For a logical fragment of \MSO equivalent to (\ie, recognizing the same languages as) a congruence class, we give sufficient conditions under which this equivalence carries over to transductions. To the best of our knowledge these conditions are satisfied by all logical fragments which are equivalent to a congruence class and which have access to the linear order predicate.
In particular, this gives a way to decide if a given rational function is definable in \FO or in \FOd.

An entire section is devoted to the particular case of aperiodic transductions, where we show that this class satisfies a stronger property than an arbitrary class: a transduction is aperiodic if and only if \emph{all} its minimal bimachines are aperiodic. This gives a decision procedure with better complexity, just by ``local'' minimization instead of computing all the minimal bimachines. Namely we show that deciding whether a transduction, given as a bimachine, is aperiodic is \pspace-complete.

\subsection*{Comparison with previous papers}
This article compiles the results of two previous conference articles~\cite{FiliotGL16, FiliotGL16-2}, as well as some new results.
The effective computation of the set of minimal bimachines of a rational function is given in~\cite{FiliotGL16} and, as a consequence, the decidability of the first-order definability problem. We have shown afterwards, in~\cite{FiliotGL16-2}, that the first-order definability problem is actually \pspace-complete.
This article details the proofs of these results and exhibits
some new interesting properties.
For instance, testing whether $f$ is a
$\var$-transduction can be achieved by testing
a unique bimachine $\bim_{f,\var}$.
The equivalence (in terms of expressiveness) between unambiguous
$\var$-transducers and functional $\var$-transducers is a new result,
as well as the decidability of the definability in \FOd and \Bsig{1}.


\tableofcontents

\newpage
\section{Rational languages and rational transductions}%
\label{sec:rational_languages}

\subsection{Rational languages}

\subsubsection{Words and languages}
An \emph{alphabet} $\Sigma$ is a finite set of symbols called \emph{letters}, a \emph{word} over $\Sigma$ is an element of the free monoid $\Sigma^ *$, whose neutral element, the \emph{empty word}, is denoted by $\epsilon$.
The \emph{length} of a word $w\in \Sigma^ *$ is denoted by $|w|$ with $|\epsilon|=0$.
For a non-empty word $w$ and two positions $1 \leq i\leq j \leq |w|$ we denote by $w[i]$ the $i$th letter of $w$ and by $w[i:j]$ the factor of $w$ beginning with the $i$th letter and ending with the $j$th letter of $w$.
For two words $u,v$, we write $u\preceq v$ if $u$ is a prefix of $v$ and in this case we denote by $u^{-1}v$ the unique word $v'$ such that $uv'=v$.
The \emph{longest common prefix} of two words $u,v$ is denoted by $u\wedge v$ and the \emph{prefix distance} between $u$ and $v$ is defined as $\dist{u,v}=|u|+|v|-2|u\wedge v|$.
A \emph{language} $L$ is a set of words and $\bigwedge L$ denotes the longest common prefix of all the words in $L$ with the convention that $\bigwedge \varnothing=\epsilon$.

\subsubsection{Finite automata}
A \emph{finite automaton} (or simply automaton) over an alphabet $\Sigma$ is a tuple $\aut=\tuple{Q,\Delta,I,F}$ where $Q$ is a finite set of \emph{states}, $\Delta\subseteq Q\times \Sigma \times Q$ is the \emph{transition relation}, and $I,F\subseteq Q$ denote the set of \emph{initial states} and the set of \emph{final states}, respectively.
A \emph{run} of $\aut$ over a word $w$ is a word $r=q_0\ldots q_{|w|}$ over $Q$ such that $(q_i,w[i+1],q_{i+1})\in \Delta$ for all $i\in \set{0,\ldots, |w|-1}$.
The run $r$ is \emph{accepting} if $q_0\in I$ and $q_{|w|}\in F$.
A word $w$ is \emph{accepted} by $\aut$ if there exists an accepting run over it, and the language  \emph{recognized} by $\aut$ is the set of words accepted by $\aut$ and is denoted by $\sem{\aut}$.
We will use the notation $p\xrightarrow w _\aut q$ (or just $p\xrightarrow w q$ when it is clear from context) to denote that there exists a run $r$ of $\aut$ over $w$ such that $r[1]=p$ and $r[|r|]=q$.
An automaton $\aut$ is called \emph{deterministic} if its set of initial states is a singleton and for any two transitions $(p,\sigma,q_1),(p,\sigma,q_2)\in \Delta$ it holds that $q_1=q_2$.
An automaton is \emph{unambiguous} if any word has at most one accepting run over it.
We call an automaton \emph{complete} if for any $p\in Q,\sigma\in \Sigma$ there exists $q\in Q$ such that $(p,\sigma,q)\in \Delta$.
A state $q$ of an automaton $\aut$ is said to be \emph{accessible} if there exists a word $w$ and an initial state $q_0$, such that $q_0\xrightarrow w _\aut q$. We say by extension that the automaton $\aut$ is \emph{accessible} if all its states are accessible.
Finally, a language is called \emph{rational} if it is recognized by an automaton.

\subsection{Algebraic characterization of rational languages}

\subsubsection{Congruences}
It is well known that rational languages are equivalently recognized by congruences of finite index. Let us define these notions.
Let $\sim$ be an equivalence relation on $\Sigma^*$.
The equivalence class of a word $w$ is denoted by $\cla w_\sim$ (or just $\cla w$ when it is clear from context).
We say that $\sim$ has finite index if the quotient $\Sigma^*/_{\sim}=\set{\cla w_\sim \mid w\in \Sigma^*}$ is finite.
Let $\sim_1$, $\sim_2$ be two equivalence relations, we say that $\sim_1$ is \emph{finer} than $\sim_2$ (or that $\sim_2$ is \emph{coarser} than $\sim_1$) if any equivalence class of $\sim_2$ is a union of equivalence classes of $\sim_1$, or equivalently for any two words $u,v$, if $u\sim_1 v$ then $u\sim_2v$. We write ${\sim_1} \finer {\sim_2}$ to denote that $\sim_1$ is finer than $\sim_2$.

An equivalence relation $\sim$ over $\Sigma^*$ is a \emph{right congruence} (resp.\ a \emph{left congruence}) if for any words $u,v\in \Sigma^ *$ and letter $\sigma\in \Sigma$ we have $u\sim v \Rightarrow u\sigma \sim v\sigma$ (resp.\ $u\sim v \Rightarrow \sigma u \sim \sigma v$).
A \emph{congruence} is defined as both a left and a right congruence.
When $\sim$ is a congruence, the quotient $\Sigma^*/_\sim$ is naturally endowed with a monoid structure, with multiplication $\cla u_\sim \cdot \cla v_\sim=\cla {uv}_\sim$ and identity $\cla \epsilon_\sim$.
We say that a congruence $\sim$ \emph{recognizes} a language $L$ if $L$ is a union of equivalence classes of $\sim$.
Let us note that the intersection of two right (resp.\ left) congruences $\sim_1,\sim_2$ over the same alphabet is also a right (resp.\ left) congruence which we denote by $\sim_1\sqcap\sim_2$.

\begin{exa}
Important examples of congruences which we use
throughout the article are: the \emph{syntactic congruence}
$\equiv_L$ of a language $L$ and the \emph{transition congruence}
$\approx_\aut$ of an automaton $\aut$ with set of states $Q$. When
$\aut$ is deterministic with initial state $q_0$, we also define
$\sim_\aut$, the \emph{right transition congruence} of $\aut$. These
relations are defined as follows, with $u,v$ ranging over all words:

\[
\begin{array}{rclclrcl}
u & \equiv_L & v & \Leftrightarrow & ( \forall x,y\in\Sigma^*, &
                                                                  xuy\in
                                                                  L &
    \Leftrightarrow & xvy\in L ) \\

u &\approx_\aut & v & \Leftrightarrow & ( \forall p,q\in Q, &
                                   p\xrightarrow{u}_\aut q
                                   & \Leftrightarrow &
                                                       p\xrightarrow{v}_\aut q ) \\
u &\sim_\aut & v & \Leftrightarrow & (\forall q\in Q, &
                                   q_0\xrightarrow{u}_\aut q
                                   & \Leftrightarrow & q_0\xrightarrow{v}_\aut q )
\end{array}
\]
If $\aut$ recognizes a language $L$, then its transition congruence $\approx_\aut$ recognizes the same language.
The relation $\equiv_L$ recognizes the language $L$ and is the coarsest among congruences which do so. From the two previous remarks, this well-known fact follows: a language is rational if and only if its syntactic congruence has finite index.

Let $\aut$ be a complete and deterministic automaton, then an equivalence class $\cla w_{\sim_\aut}$ can be identified with the state of $\aut$ reached by reading the word $w$. In the following we will often make this identification, implicitly assuming that $\aut$ is complete (an automaton can be made complete in \ptime), and will write $\cla w_{\aut}$ rather than $\cla w_{\sim_\aut}$ to simplify notations. Furthermore, for two deterministic automata $\aut_1,\aut_2$ such that ${\sim_{\aut_1}}\finer {\sim_{\aut_2}}$, we will say by extension that $\aut_1$ is finer than $\aut_2$ and write $\aut_1\finer\aut_2$. For example the minimal automaton of a language $L$ is the coarsest deterministic automaton recognizing $L$
(up to isomorphism).
\end{exa}

\subsubsection{Congruence classes}

A \emph{congruence class} (class for short) $\var$ associates to any finite alphabet $\Sigma$, a set $\var(\Sigma)$ of congruences of finite index over $\Sigma$, such that $\var(\Sigma)$ is (1) closed under intersection and (2) closed under taking coarser congruences.
Note that we will often abuse the notation and write $\var$ instead of $\var(\Sigma)$.

An automaton is called a \emph{$\var$-automaton} if its transition congruence is in $\var$.
A language is a \emph{$\var$-language} if it is recognized by a $\var$-automaton.
Note that since a class is stable by taking coarser congruences, a language is a $\var$-language if and only if its syntactic congruence is in $\var$. Let $\mathcal L(\var)$ denote the set of $\var$-languages.

\begin{rem}
Recognizability by a congruence of finite index is equivalent to the notion of recognizability by a \emph{stamp} (\ie, a surjective morphism from a free monoid to a finite monoid).
Indeed, a congruence $\sim$ over an alphabet $\Sigma$ yields a natural (surjective) morphism $\cla \cdot _\sim:\Sigma^*\rightarrow \Sigma^*/_\sim$ which of course recognizes the same language.
Conversely, given a morphism $\mu:\Sigma^*\rightarrow M$ with $M$ a finite monoid, one can define the congruence $\sim_M$ by $u\sim_M v$ if $\mu(u)=\mu(v)$.
In~\cite{PinS05}, the authors define $\mathcal C$-varieties of stamps (a generalization of the notion of monoid varieties), which are in particular congruence classes as we define them in this paper.
We choose to consider this notion of congruence classes simply because our results hold in this more general framework.

Note that in~\cite{FiliotGL16} we had chosen the term of \emph{congruence variety} instead of \emph{congruence class} which was ill-suited, as was kindly pointed out by Jean-\'Eric Pin, since the term variety stems from an equational theory, which does not exist in a context as general as the one of congruence classes.
\end{rem}

\subsubsection{Definability problem and decidable classes}

Given a set of languages $\mathcal V$, the \emph{$\mathcal V$-definability} problem asks whether a language $L$, given by an automaton, belongs to $\mathcal V$.

A congruence of finite index over an alphabet $\Sigma$ can be given by
a morphism $\mu:\Sigma^ *\rightarrow M$, with $M$ a finite monoid (two
words are equivalent if they have the same image by $\mu$).
The morphism $\mu$ can itself be given explicitly by a function $m:\Sigma\rightarrow M$. In the following decision problems we assume that congruences are given that way.

The \emph{$\var$-membership problem} asks if a congruence of finite
index, given as a morphism, is in $\var$.
In particular, the $\mathcal L(\var)$-definability problem reduces to the $\var$-membership problem through the syntactic congruence.
A class is called \emph{decidable} if its membership problem is decidable.

\begin{exas}
We give several examples of decidable congruence classes.
\begin{itemize}
\item The set $\fin$ of all congruences of finite index is of course a class.

\item A simple example of a congruence class is the class $\id$ of \emph{idempotent} congruences.
A congruence $\sim$ is called idempotent if for any word $w$, $w\sim w^2$, and this property is indeed stable by intersection and taking coarser congruences.

\item A central example of congruence class in this paper is the class $\ap$ of \emph{aperiodic} congruences.
A congruence $\sim$ is aperiodic if there exists an integer $n$ such that for any word $w$, $w^n\sim w^{n+1}$.
Again, one can easily check that this property is stable by intersection and taking coarser congruences.
It is shown in~\cite{DiekertG08} that the aperiodicity problem for a language, given as an automaton, is \pspace-complete.
Moreover, even if the given  automaton is deterministic, this problem remains \pspace-hard~\cite{ChoH91}.

\item Let us also mention the congruence class $\da$ which was shown to recognize the languages definable in first-order logic with two variables~\cite{TherienW98}. A congruence of finite index $\sim$ is in $\da$ if there exists an integer $n$ such that for any words $u,v,w$, ${(uvw)}^{n}v(uvw^n)\sim{(uvw)}^n$.

\end{itemize}

\end{exas}

\subsection{Rational transductions}

\subsubsection{Transductions and finite transducers}

A \emph{transduction} over an alphabet $\Sigma$ is a partial function\footnote{In this paper transductions are functions although this term may refer to any relation in the literature.} $f:\Sigma^*\rightarrow \Sigma^*$, and its domain is denoted by $\dom(f)$.

A \emph{finite transducer}\footnote{This type of transducer is sometimes called \emph{real-time}~\cite{Sakarovitch09} since in general a transition of a transducer may be labelled by any word.} (or simply \emph{transducer}) over $\Sigma$ is a tuple $\trans=\tuple{\aut,\out,\init,\final}$ where $\aut=\tuple{Q,\Delta,I,F}$ is the \emph{underlying automaton} of $\trans$, $\out:\Delta\rightarrow\Sigma^*$ is the \emph{output function}, $\init:I\rightarrow\Sigma^*$ is the \emph{initial output function} and $\final:F\rightarrow\Sigma^*$ is the \emph{final output function}.
Let $u$ be a word on which there exists a run $q_0\ldots q_{|u|}$ of $\aut$, and let $v=\out(q_0,u[1],q_1)\cdots\out(q_{|u|-1},u[|u|],q_{|u|})$, then we write $q_0\xrightarrow {u\mid v} _\trans q_{|u|}$ to denote the existence of such a run.
If $q_0\in I$ and $q_{|u|}\in F$, let $w=\init(q_0)v\final(q_{|u|})$, then we say that the pair $(u,w)$ is \emph{realized} by $\trans$.
We denote by $\sem \trans$ the set of pairs realized by $\trans$.

A transducer is called \emph{functional} if it realizes a transduction (\ie, a partial function), this property is decidable in \ptime (see \eg~\cite{BerstelB79}) and in this case we denote $(u,v)\in \sem\trans$ by $\sem\trans(u)=v$.
A transducer $\trans$ is called \emph{unambiguous} (resp.\ \emph{sequential}) if its underlying automaton is unambiguous (resp.\ deterministic) and in both cases $\trans$ is functional.
It is known that any functional transducer is equivalent to some unambiguous transducer (see \eg~\cite{BerstelB79}).
Finally, a transduction is called \emph{rational} (resp.\ \emph{sequential}) if it is realized by a functional (resp.\ sequential) transducer.

\subsubsection{$\var$-transducers}

Let $\var$ be a congruence class.
A \emph{$\var$-transducer} is a transducer whose underlying automaton is a $\var$-automaton.
A transduction is called \emph{$\var$-rational} (resp.\ \emph{$\var$-sequential}) if it is realized by a functional (resp.\ sequential) $\var$-transducer.
Note that a $\var$-rational transduction is sometimes called a $\var$-transduction for short.

\section{Algebraic characterization of sequential transductions}%
\label{sec:algebra_sequential}

Sequential transductions can be characterized, as it was shown in~\cite{Choffrut03}, by a syntactic congruence which, like for automata, yields a unique minimal underlying machine.
Given a congruence class $\var$ we thus prove that a sequential transduction is $\var$-sequential if and only if its minimal transducer is a $\var$-transducer, which provides a decision procedure for $\var$-sequentiality in the case of a decidable congruence class.
The second result of this section is that a sequential $\ap$-transduction is also $\ap$-sequential, as determinization preserves aperiodicity.
These results are depicted in Figure~\ref{fig:sequential}.
\begin{figure}

  \begin{tikzpicture}[baseline=0, inner sep=0, outer sep=0, minimum size=0pt, scale=0.32]
  \tikzstyle{cross} = [minimum size=4pt, path picture={
      \draw[black] (path picture bounding box.south east) -- (path picture bounding box.north west) (path picture bounding box.south west) -- (path picture bounding box.north east);
  }]
  \tikzstyle{projection} = [->, >=stealth, shorten >=1pt, thick, rounded corners=5]
  \tikzstyle{pointzone} = [->, >=stealth, shorten >=1pt, dash pattern=on5pt off3pt, looseness=1]
  \tikzstyle{intersections} = [pattern=north east lines, pattern color=gray!65];

\begin{scope}

  \draw (2,12) node {transductions};
  \draw (2, 2) node {transducers};
  \draw (0,7) -- (45,7);

  \newcommand\Vzone{(16,12) ellipse ( 8cm and 3.5cm)}
  \newcommand\Seqzone{(26,12) ellipse (12cm and 4cm)}
  \newcommand\Azone{(36,12) ellipse ( 8cm and 3.5cm)}
  \begin{scope} 
    \clip \Seqzone;
    \fill[intersections] \Vzone;
    \fill[intersections] \Azone;
  \end{scope}
  \draw \Vzone;   \draw (13.5,14.5) node {$\var$};
  \draw \Seqzone; \draw (  26,14.8) node {sequential};
  \draw \Azone;   \draw (38.5,14.5) node {$\ap$};

  \draw (14,2) ellipse ( 6cm and 3.5cm); \draw (13.5,-.5) node {$\var$};
  \draw (26,2) ellipse (12cm and   4cm); \draw (26,-.8) node {sequential};
  \draw (36,2) ellipse ( 8cm and 3.5cm); \draw (38.5,-.5) node {$\ap$};

  \draw [fill=white] (18,12) ellipse (2.8cm and 2.5cm); \draw (18,13.3) node {$\var$-seq.};
  \draw (18,11) node [cross] (nodef1)    {}; \draw (18,12) node {$f_1$};
  \draw (17.5, 3) node [cross] (nodeminf1) {}; \draw (17.5, 2) node {$\min(f_1)$};
  \draw [projection] (nodef1) -- (nodeminf1);
  \draw (22,11) node [cross] (nodef2)    {}; \draw (22,12) node [fill=white] {$f_2$};
  \draw (23, 3) node [cross] (nodeminf2) {}; \draw (23, 2) node {$\min(f_2)$};
  \draw (12, 3) node [cross] (nodeT2)    {}; \draw (12, 2) node {$T_2$};
  \draw [projection] (nodef2) -- (nodeminf2);
  \draw [projection] (nodef2) -- (nodeT2);
  \draw (13,20) node {$\not=\emptyset$ for $\var=\id$ (Prop~\ref{prop:not-V-seq})};
  \path (20,20) edge [pointzone, bend left] (21,14);
  \draw (20,7) node [fill=white] {Th.~\ref{thm:v-seq}};

  \draw [fill=white] (34,12) ellipse (2.8cm and 2.5cm); \draw (34,13.3) node {$\ap$-seq.};
  \draw (34,11) node [cross] (nodef3) {}; \draw (34,12) node {$f_3$};
  \draw (40, 3) node [cross] (nodeT3) {}; \draw (40, 2) node {$T_3$};
  \draw [projection] (nodef3) -- (nodeT3);
  \draw (34, 3) node [cross] (nodedetT3) {}; \draw (34, 2) node {$\mathrm{det}(T_3)$};
  \draw [->, >=stealth, shorten >=1pt] (nodeT3) -- (nodedetT3) node[above=1mm,midway] {Th.~\ref{thm:a-seq}};
  \draw (35.5,20) node {$=\emptyset$ (Theorem~\ref{thm:a-seq})};
  \path (30,20) edge [pointzone, bend right] (31,14);

\end{scope}
\end{tikzpicture}


  \caption{Situation for sequential transductions.\label{fig:sequential}}
\end{figure}
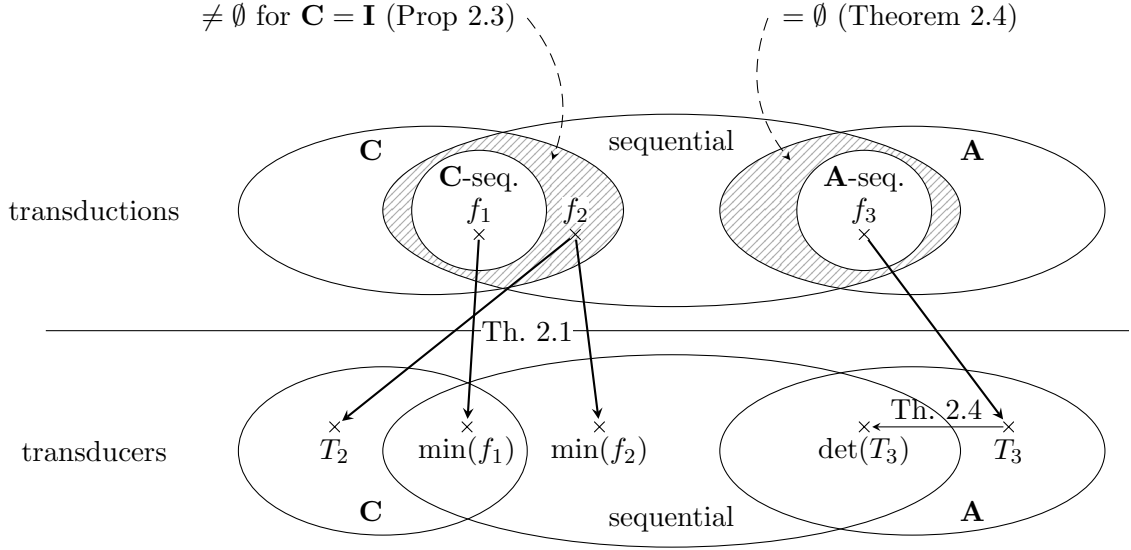

\subsection{Minimization of sequential transducers}%
\label{subsec:min-seq}

We describe the minimal sequential transducer given in~\cite{Choffrut03} and refer the reader to the original paper for a proof that the obtained transducer realizes the same transduction.
We then show that the underlying automaton of the minimal transducer is minimal in the strong algebraic sense that it is the coarsest among all the underlying automata of sequential transducers realizing the same transduction.
We will thus obtain a decision procedure for $\var$-sequentiality of
sequential transductions:

\begin{thm}%
\label{thm:v-seq}
Let $\var$ be a decidable congruence class.
It is decidable whether a sequential transduction, given by a transducer, is $\var$-sequential.
\end{thm}

\subsubsection*{Minimal transducer}

Let $f$ be a transduction and let us define $\trans_f=\tuple{\aut_f,\out_f,\init_f,\final_f}$ the minimal sequential transducer realizing $f$ with $\aut_f=\tuple{Q_f,\Delta_f,I_f,F_f}$.
The main idea of the procedure is to output the letters as soon as possible, and then define a right congruence which states that two words are equivalent if the outputs are the same for any continuation.

For this we need to define a new transduction $\widehat{f}: \Sigma^*\rightarrow \Sigma^*$ by: for any word $u$, $\widehat{f}(u)=\bigwedge \set{f(uw)\mid\ w\in u^{-1}\dom(f)}$, which captures this ``as soon as possible'' idea by outputting the longest common prefix of the images of all words beginning with $u$.
The \emph{syntactic congruence} of $f$ is defined by: $u \sim_f v$ if (1) for any $w\in \Sigma^*$, $uw\in \dom(f) \Leftrightarrow vw\in \dom(f)$ and (2) for any $w\in u^{-1}\dom(f)$, ${\widehat{f}(u)}^{-1} f(uw)={\widehat{f}(v)}^{-1} f(vw)$ for any words $u,v$.
The first condition only ensures that the congruence recognizes the domain of the transduction, while the second states that the output due to the continuation $w$ is the same after reading $u$ and $v$.

The automaton $\aut_f$ is defined naturally from the congruence $\sim_f$:
\begin{itemize}
\item $Q_f=\Sigma^*/_{\sim_f}$
\item $\Delta_f=\set{\tuple{\cla u,\sigma,\cla {u\sigma}}\mid\ u\in \Sigma^*,\sigma\in\Sigma}$
\item $I_f=\set{\cla \epsilon}$
\item $F_f=\set{\cla u\mid\ u\in \dom(f)}$
\end{itemize}
The outputs are defined using the transduction $\widehat{f}$.
\begin{itemize}
\item $\out\tuple{\cla u,\sigma,\cla {u\sigma}}={\widehat{f}(u)}^{-1}\widehat{f}(u\sigma)$
\item $\init(\cla \epsilon)=\widehat{f}(\epsilon)$
\item $\final(\cla u)= {\widehat{f}(u)}^{-1} f(u)$ for $u\in \dom(f)$
\end{itemize}

\noindent
Before proving Theorem~\ref{thm:v-seq} we need to check that taking a coarser automaton preserves the congruence class.

\begin{prop}%
\label{prop:finer-seq}
Let $\var$ be a congruence class. Let $\aut_1,\aut_2$ be two deterministic automata such that $\aut_1\finer\aut_2$ and $\aut_2$ is accessible.
If $\aut_1$ is a $\var$-automaton then so is $\aut_2$.
\end{prop}

\begin{proof}
Let $\aut_1\finer \aut_2$ such that $\aut_1$ is a $\var$-automaton and $\aut_2$ is accessible.
We only need to show that ${\approx_{\aut_1}} \finer {\approx_{\aut_2}}$ since a congruence class is stable by taking coarser congruences.
Let $u\approx_{\aut_1} v$, then for any word $w$, $wu\approx_{\aut_1}wv$ and in particular $wu\sim_{\aut_1}wv$, hence $wu\sim_{\aut_2}wv$. Since $\aut_2$ is accessible, we have $u\approx_{\aut_2} v$.
\end{proof}

\begin{proof}[Proof of Theorem~\ref{thm:v-seq}]
Let $\var$ be a congruence class.
Let us show that a transduction $f$ is $\var$-sequential if and only if $\trans_f$ is a $\var$-transducer.
This is enough to prove the theorem since $\trans_f$ can be computed, in \ptime, from any sequential transducer realizing $f$, according to~\cite{Choffrut03}.

The ``if'' direction is trivial since in particular $\trans_f$ is a sequential $\var$-transducer realizing $f$.
Now, let $\trans=(\aut,\out,\init,\final)$ be a sequential $\var$-transducer realizing $f$, we want to show that $\aut_f$ is a $\var$-automaton.
Since $\aut_f$ is by definition accessible, we only need to show, according to Proposition~\ref{prop:finer-seq}, that $\aut\finer\aut_f$.

Let $u\sim_{\aut} v$, let us show that $u\sim_f v$. Since $\aut$ recognizes $\dom(f)$ we already have that for any word $w$, $uw\in \dom(f) \Leftrightarrow vw\in \dom(f)$.
Let $w$ be, if it exists, a word such that $uw\in \dom(f)$ and let $p\xrightarrow{u|x}_{\trans}q \xrightarrow{w|z}_{\trans} r$ denote the corresponding accepting run, and similarly for $v$: $p\xrightarrow{v|y}_{\trans}q \xrightarrow{w|z}_{\trans} r$.
We have $\widehat{f} (u)=\init(p)x z'$ and $\widehat{f} (v)=\init(p)y z'$ with $z'\preceq z\final(r)$ the longest common prefix of all outputs from state $q$ to an accepting state.
Finally,
\begin{align*}
{\widehat{f}(u)}^{-1} f(uw) &=\tuple{\init(p)x z'}^{-1}\init(p)x z\final(r)\\
&=z'^{-1}z\final(r)\\
&=\tuple{\init(p)y z'}^{-1}\init(p)y z\final(r)\\
&={\widehat{f}(v)}^{-1} f(vw)
 \qedhere
\end{align*}
\end{proof}

\subsection{Determinization preserves aperiodicity}
We have shown how to decide if a transduction is $\var$-sequential.
One could wonder if a sequential transduction can be $\var$-rational but not $\var$-sequential.
We answer by the affirmative, and show that for the case of aperiodicity this cannot happen.

\begin{prop}%
\label{prop:not-V-seq}
There exists a congruence class $\var$ and a sequential $\var$-transduction which is not $\var$-sequential.
\end{prop}

\begin{proof}

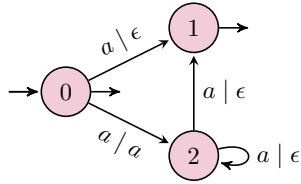
\begin{figure}[t]
\centering

\begin{tikzpicture}[->,>=stealth',shorten >=1pt,auto,node distance=1.6cm,
 semithick,scale=0.85,every node/.style={scale=0.85}]
 \tikzstyle{every state}=[fill=purple!20!white,minimum size=2em]
\node[initial,state, initial by arrow, accepting by
  arrow,initial text=,accepting text=](q0) at (0,0) {0} ;
\node[accepting by arrow, state] (q1) at (2,1) {1} ;
\node[state] (q2) at (2,-1) {2} ;

\draw[->,>=stealth] (q0)--(q1) node[midway,above,sloped] {$a\mid \epsilon$};
\draw[->,>=stealth] (q0)--(q2) node[midway,below,sloped] {$a\mid a$};
\draw[->,>=stealth] (q2)--(q1) node[midway,right] {$a\mid \epsilon$};
\draw[] (q2) edge[loop right] node[midway, right]{$a\mid\epsilon$} (q2);
\end{tikzpicture}

\caption{$\id$-transducer.}\label{i-trans}
\end{figure}


 Let us consider the congruence class $\id$ of idempotent congruences.
We give in Figure~\ref{i-trans} an example of an unambiguous $\id$-transducer which realizes a sequential transduction $f$.
Let us show that $f$ is not $\id$-sequential.
We assume by contradiction that there exists a sequential $\id$-transducer realizing $f$.
Let $p\xrightarrow{a|\epsilon}q$ be the corresponding accepting run over $a$.
Since $f(a) = \epsilon$ we have $\init(p)=\final(q)=\epsilon$.
Moreover, since $f(aa)=a$ and $a\sim aa$, we have $q\xrightarrow{a|a}q$. Hence $f(aaa)=aa$ which yields a contradiction.
\end{proof}

In this section we show however that this is not the case for the congruence class $\ap$ of aperiodic congruences.
This fact is shown by exhibiting a determinization algorithm which, when it takes as an input an $\ap$-transducer realizing a sequential transduction, yields a sequential $\ap$-transducer.
As it was shown in~\cite{BealC02}, a transducer defining a sequential transduction can be determinized using a modified subset construction  where the automaton has to store delays between outputs in its states.
This procedure does not preserve the congruence class in general but does preserve aperiodicity, which allows us to state the following theorem:
\begin{thm}%
\label{thm:a-seq}
A sequential transduction is $\ap$-sequential if and only if it is $\ap$-rational.
\end{thm}

\subsubsection*{Determinization algorithm}
Let us give the determinization algorithm from~\cite{BealC02}.
The idea of the algorithm is similar to the subset construction for automata, but taking the outputs into account:
on a transition from a subset to another, the output is the longest common prefix for all the possible transitions.
The rest of the outputs have to be remembered in the states themselves.
Since not all rational transductions are sequential, the algorithm may not terminate.
However it is shown in~\cite{BealC02} that if the transduction is sequential, the algorithm does terminate.

Now let us describe the algorithm:
Let $\trans=(\aut,\out,\init,\final)$ be a transducer realizing a transduction $f$, with underlying automaton $\aut=(Q,\Delta,I,F)$.
We give a construction of $\trans'=(\aut',\out',\init',\final')$ a transducer realizing $f$, with $\aut'=(Q',\Delta',\set{S_0},F')$ being deterministic.
Let $j=\bigwedge\{\init(q)\mid\ q\in I\}$.
Then $S_0=\{(q,w)\mid\ q\in I\text{ and } \init(q)=jw\}$.
From the initial state we build the states and the transitions of $\aut'$ inductively.
Let $S_1$ be a state already constructed and let $\sigma\in \Sigma$.
We define $R_2=\{(p,vu)\mid\ (q,v)\in S_1 \text{ and } q\xrightarrow{\sigma\mid u}_\trans p\}$.
Let $s=\bigwedge\{w\mid\ (q,w)\in R_2\}$.
Then we define a new state of $Q'$, $S_2=\{(q,w)\mid\ (q,sw)\in R_2\}$ and add the transition to $\Delta'$ and the output of the transition:
$S_1\xrightarrow{\sigma\mid s}_{\trans'}S_2$.
Assuming that $f$ is sequential, the construction must terminate, and we only have left to describe:
\begin{itemize}
\item $\init'(S_0)=j$
\item $F'=\{S\in Q'\mid\ \exists q\ \in F, w\in \Sigma^*,\ (q,w)\in S\}$
\item $\final'(S)=w\final(q)$ such that $q\in F$ and $(q,w)\in S$
\end{itemize}
The definition of $\final'$ may seem ambiguous but it is well-defined due to $f$ being functional.


\begin{proof}[Proof of Theorem~\ref{thm:a-seq}]

Let $\trans=\tuple{\aut,\init,\out,\final}$, with $\aut=\tuple{Q,\Delta,I,F}$, be an $\ap$-transducer realizing a sequential function $f:\Sigma^*\rightarrow \Sigma^*$. Let $\trans'=\tuple{\aut',\init',\out',\final'}$, with $\aut'=\tuple{Q',\Delta',{S_0},F'}$,  be the sequential transducer obtained from $\trans$ by the powerset construction with delays described above. We will show that $\aut'$ is \emph{counter-free}, \ie, for any state $S$, any word $u$, any integer $n>0$, if $S\xrightarrow{u^n}_{\aut'} S$ then $S\xrightarrow{u}_{\aut'} S$.
This condition is sufficient for an automaton to be aperiodic (see \eg~\cite{DiekertG08}). $\aut$ is aperiodic so there is an integer $n$ such that $\forall u\in \Sigma^*$, $u^n\approx_\aut u^{n+1}$.

Let $u\in \Sigma^+$ be a word, let $k$ be a positive integer and let
\[
    R_0\xrightarrow{u|\alpha_0}_{\trans'} R_1 \xrightarrow{u|\alpha_1}_{\trans'} \ldots \xrightarrow{u|\alpha_{k-2}}_{\trans'} R_{k-1} \xrightarrow {u|\alpha_{k-1}}_{\trans'} R_0
\]
denote a counter in $\trans'$. Let us assume that $k$ is the size of the smallest such counter, which means that all $R_j$s are pairwise distinct, we want to show $k=1$.

Let $\Gamma_0:=\alpha_0\cdots \alpha_{k-1}$, for $1 \leq j < k$ let $\Gamma_j:= \alpha_j\cdots\alpha_{k-1}\alpha_0\cdots \alpha_{j-1}$ and let us note that $\Gamma_j\alpha_j=\alpha_j\Gamma_{j+1\mod k}$.
Let $R=\set{q_1, \ldots,q_m}$ denote the states appearing in $R_0$. For $0\leq j < k$, the states of $R_j$ are exactly the states which can be reached in $\aut$ from some state of $R_0$ by reading $u^{kn+j}\approx_\aut u^{kn}$. This means that the states of $R_j$ are the same as the states of $R_0$, namely $q_1, \ldots,q_m$.
Thus let $R_j=\set{(q_1,\beta_{1,j}), \ldots, (q_m, \beta_{m,j})}$.

Let $i,i' \in \set{1,\ldots,m} $, and let $q_{i} \xrightarrow{u|\alpha_{i,i'}}_\trans q_{i'}$ denote a run in $\trans$ when it exists. By definition of $\trans'$, we have for any $0\leq j < k$: \[\beta_{i,j}\alpha_{i,i'}=\alpha_{j}\beta_{i',j+1 \mod k}\]
Let $q_{i_0}\xrightarrow{u}_\aut q_{i_1} \ldots q_{i_{t-1}} \xrightarrow{u}_\aut q_{i_t}$ such that $t=ks$ is a multiple of $k$. Thus we obtain for any $0\leq j,j' < k$:
\[ \begin{array}{rcl}
\beta_{i_0,j}\alpha_{i_0,i_1}\cdots \alpha_{i_{t-1},i_t} &=& \Gamma_{j}^s \beta_{i_t,j} \\
\beta_{i_0,j'}\alpha_{i_0,i_1}\cdots \alpha_{i_{t-1},i_t} &=& \Gamma_{j'}^s \beta_{i_t,j'} \\
\end{array}\]


Since $Q$ is finite there must be a state $q_l\in R$, such that $q_l$ loops by reading a power of $u$, meaning that there is an integer $t$ such that $q_l\xrightarrow{u^t}_\aut q_l$, and we call such a state a \emph{looping state}. For a large enough $t$ we can assume by aperiodicity that $t$ is of the form $t=ks+1$. Let $l,i_1,\ldots,i_{t-1},l$ denote the state indices of the previous run from $q_l$ to $q_l$ over $u^t$. Let $\Phi:= \alpha_{l,i_1}\cdots \alpha_{i_{t-1},l}$. We have for $0\leq j < k$:
\[ \begin{array}{rclr}
\beta_{l,j}\Phi &=&\Gamma_{j}^s\alpha_j\beta_{l,j+1}& (1) \\
\beta_{l,j}\Phi^k &=&\Gamma_{j}^{ks+1}\beta_{l,j} & (2)\\
\end{array}
\]
From (2) we have $|\Phi|=(ks+1)\frac{|\Gamma_0|}{k}$. From (1) we thus obtain: $ |\beta_{l,j+1}|-|\beta_{l,j}|=\frac {|\Gamma_0|}{k}-|\alpha_j| $.
Notice that this holds for any looping state $q_l$ but the difference does not depend on the state itself.

Let us now consider $q_i$, a state which is not necessarily a looping state.
Any state must be reachable from some looping state, since all states in $R$ can be reached from some state of $R$ by an arbitrarily large power of $u$. Let $q_l$ be a looping state which can reach $q_i$ by a run over $u^{ks}$. Again let $l,i_1',\ldots,i_{ks-1}',i$ denote the sequence of indices of such a run and let $\Psi:=\alpha_{l,i_1'}\cdots \alpha_{i_{ks-1}',i}$. We have for $0\leq j < k$:
\[ \begin{array}{rclr}
\beta_{l,j}\Psi &=&\Gamma_{j}^{s}\beta_{i,j} & (3)\\
\beta_{l,j+1}\Psi &=&\Gamma_{j+1}^{s}\beta_{i,j+1} & (4)\\
\end{array}
\]
Taking the lengths of words of (4) and (5), and taking the difference between the two equalities we obtain:
$|\beta_{i,j+1}|-|\beta_{i,j}|= |\beta_{l,j+1}|-|\beta_{l,j}|=\frac {|\Gamma_0|}{k}-|\alpha_j| $ which again does not depend on $i$. Thus we obtain that for any state $q_i\in R$, looping or not, $|\beta_{i,j+1}|-|\beta_{i,j}|=\frac {|\Gamma_0|}{k}-|\alpha_j|$.

If we assume that $\Gamma_0=\epsilon$, then in particular $|\beta_{i,j+1}|=|\beta_{i,j}|$ for any state $q_i$. From (1), we have that $\beta_{l,j}=\beta_{l,j+1}$ for any looping state $q_l$. Then combining (4) and (5) we obtain $\beta_{i,j}=\beta_{i,j+1}$ for any state $q_i$. Hence all $R_j$s are identical which means that $k=1$.

Let us now assume that $\Gamma_0\neq \epsilon$. Since equations (4) and (5) can have an arbitrarily large common suffix, we know that for any state $q_i$, either $\beta_{i,j}$ is a suffix of $\beta_{i,j+1}$ or \emph{vice versa}.
Note that whether $\beta_{i,j}$ or$\beta_{i,j+1}$ is a suffix of the other does \emph{not} depend on $i$ since $|\beta_{i,j+1}|-|\beta_{i,j}|=\frac {|\Gamma_0|}{k}-|\alpha_j|$, and furthermore the size of $\gamma_i$ does not depend on $i$ either. If $\beta_{i,j}=\gamma_i\beta_{i,j+1}$, since (4) and (5) can have an arbitrarily large common suffix, we have that $\gamma_i$ is a suffix of $\Gamma_{j+1}^{ks}$ which does not depend on $i$.
Hence $\gamma_i$ is a common prefix of $\beta_{i',j}$ for all $i'\in \set{1,\ldots,m}$, which means that $\gamma_i=\epsilon$ by definition of $\trans'$. Thus all $R_j$s are equal which means that $k=1$. Similarly, if $\beta_{i,j}$ is a suffix of $\beta_{i,j+1}$, then $\gamma_i$ is a suffix of $\Gamma_{j}^{ks}$, and with the same reasoning, we conclude that $k=1$.
\end{proof}


\section{Bimachines}%
\label{sec:bimachines}

Bimachines are a model of computation as expressive as (functional) transducers, that was introduced by~\cite{Schutzenberger61} and further studied (and named) by~\cite{Eilenberg74}.
One of the main features of bimachines is their completely deterministic nature.
In order to express all the rational transductions a bimachine needs, as its name suggests, two automata:
A right automaton, which reads words deterministically from right to left and a left automaton (which is just a deterministic automaton).
The roles of the two automata are completely symmetrical, however the right automaton can be seen intuitively as a regular look-ahead for the left automaton.

Using a result from~\cite{ReutenauerS91} we show that for a given bimachine, one can minimize the left automaton with respect to the right one and \emph{vice versa}.
We give a \ptime algorithm for bimachine minimization (in the vein of Moore's DFA minimization algorithm), but underline the fact that a given transduction does not have a unique minimal bimachine, in general.

Finally, for a rational transduction $f$, we describe the canonical bimachine from~\cite{ReutenauerS91} (\ie, it does not depend on the description of $f$), which relies on the existence of a canonical right automaton.
Intuitively, this right automaton represents the ``minimal'' look-ahead information needed to realize the transduction sequentially.

\subsection{Bimachines and transductions}

\subsubsection{Right automaton}
Formally a \emph{right automaton} over an alphabet $\Sigma$ is an automaton $\raut=(Q,\Delta,I,F)$ such that $I$ is a singleton and its transitions are backward deterministic, meaning that for any two transitions $(p_1,\sigma,q),(p_2,\sigma,q)\in \Delta$ it holds that $p_1=p_2$.
The only difference with the classical notion of automaton lies in the definition of accepting runs.
A run $r$ over a right automaton is called \emph{accepting} if $r[1]$ is final and $r[|r|]$ is initial.
Therefore a right automaton can be thought of as reading words from right to left, deterministically.
A run $r$ of $\raut$ over the word $w$ will be denoted by $r[1]\xleftarrow{w}_{\raut}r[|r|]$ to emphasize that $\raut$ is a right automaton, and a transition $(p,\sigma,q)$ will be depicted by an arrow from $q$ to $p$.

\subsubsection{Left congruence}
The \emph{left transition congruence} associated with a right automaton $\raut=(Q,\Delta,\set{r_0},F)$ is defined by $u\sim_\raut v$ if $\forall r\in Q,\ r\xleftarrow{u}_{\raut}r_0 \Leftrightarrow r\xleftarrow{v}_{\raut}r_0$.
Exactly like  for left automata, we assume that $\raut$ is (co-)complete and identify $\cla u _{\sim_\raut}$ (often denoted by $\cla u _{\raut}$) with the unique state $r\in Q$ such that $r\xleftarrow{u}_{\raut}r_0$. 
We also say that a right automaton $\raut_1$ is finer than $\raut_2$ (denoted by $\raut_1\finer\raut_2$) if ${\sim_{\raut_1}}\finer{\sim_{\raut_2}}$.

\subsubsection{Bimachine}
A \emph{bimachine} over an alphabet $\Sigma$ is a tuple $\bim=\tuple{\laut,\raut,\bout,\lfinal,\rfinal}$ where $\laut=\tuple{Q_\laut,\Delta_\laut,\set{l_0},F_\laut}$ is a left (\ie, deterministic) automaton, $\raut=\tuple{Q_\raut,\Delta_\raut,\set{r_0},F_\raut}$ is a right automaton, $\bout:Q_\laut\times\Sigma\times Q_\raut\rightarrow \Sigma^ *$ is the \emph{output function}, $\lfinal:F_\raut \rightarrow \Sigma^*$ is the \emph{left final function} and $\rfinal:F_\laut \rightarrow \Sigma^*$ is the \emph{right final function}.
The two automata $\laut$ and $\raut$ are required to recognize the same language.

We extend naturally the function $\bout$ to $Q_\laut\times\Sigma^*\times Q_\raut$ as follows:
for all $l\in Q_\laut,r\in Q_\raut$, $\bout(l,\epsilon,r)=\epsilon$ and for all $u,v\in \Sigma^*$, $l'\in Q_\laut,r'\in Q_\raut$ such that $l\xrightarrow u_\laut l'$ and $r' \xleftarrow v _\raut r$ then $\bout(l,uv,r)=\bout(l,u,r')\bout(l',v,r)$.

We define  $\sem\bim$ the transduction \emph{realized} by $\bim$ such that $\dom( \sem\bim)=\sem \laut=\sem \raut$, and for any word $u\in \sem \laut$ if $l$ and $r$ are the final states of the runs over $u$ of $\laut$ and $\raut$ respectively, then $\sem\bim(u)=\lfinal(r)\bout(l_0,u,r_0)\rfinal(l)$.

\begin{exa}
Let us give an example of a bimachine $\bim=\tuple{\laut,\raut,\bout,\lfinal,\rfinal}$ realizing the function $\mathsf{swap}$ over the alphabet $\set{a,b}$, which swaps the first and last letters of a word.
Formally, $\mathsf{swap}(\sigma w\tau)=\tau w\sigma$ for
$\sigma,\tau\in\Sigma$ and $w\in\Sigma^ *$, and $\mathsf{swap}(w)=w$
if $|w|<2$. The automata of $\bim$ are given in Figure~\ref{fig:ex-bim}; the role of the left automaton is to remember the first letter of the word and the role of the right automaton is to remember the last letter.
The functions $\lfinal$ and $\rfinal$ both map any state to $\epsilon$, and the output function $\bout$ maps any triple of the form $(l_0,\sigma,r_\tau)$ or $(l_\tau,\sigma,r_0)$ to $\tau$ and any other triple $(l,\sigma,r)$ to $\sigma$, for any letters $\sigma,\tau$.
An execution of $\bim$ over the word $aabb$ is illustrated on Figure~\ref{fig:ex-bim}.
\end{exa}

\begin{figure}


\begin{center}
\begin{tikzpicture}[->,>=stealth',shorten >=1pt,auto,node distance=1.6cm,
 semithick,scale=0.76,every node/.style={scale=0.78}]
 \tikzstyle{every state}=[fill=lightgray!30!white,text=black,minimum size=2em]

\node[initial,state, initial by arrow, accepting by
  arrow,initial text=,accepting text=] (l0) at (0,0) {$l_0$} ;
\node[accepting by arrow, state] (la) at (2,1) {$l_a$} ;
\node[accepting by arrow, state] (lb) at (2,-1) {$l_b$} ;

\draw[->,>=stealth] (l0)--(la) node[midway,above,sloped] {$a$};
\draw[->,>=stealth] (l0)--(lb) node[midway,below,sloped] {$b$};
\draw[] (la) edge[loop above] node[midway, above]{$a,b$} (la);
\draw[] (lb) edge[loop above] node[midway, above]{$a,b$} (lb);


\node[initial right,state, initial by arrow, accepting left,initial text=,accepting text=] (l0) at (16,0) {$r_0$} ;
\node[accepting left, state] (la) at (14,1) {$r_a$} ;
\node[accepting left, state] (lb) at (14,-1) {$r_b$} ;
\draw[->,>=stealth] (l0)--(la) node[midway,above,sloped] {$a$};
\draw[->,>=stealth] (l0)--(lb) node[midway,below,sloped] {$b$};
\draw[] (la) edge[loop above] node[midway, above]{$a,b$} (la);
\draw[] (lb) edge[loop above] node[midway, above]{$a,b$} (lb);

\node[minimum size=20pt] (l0) at (4,1) {$l_0$} ;
\node[minimum size=20pt] (l1) at (6,1) {$l_a$} ;
\node[minimum size=20pt] (l2) at (8,1) {$l_a$} ;
\node[minimum size=20pt] (l3) at (10,1) {$l_a$} ;
\node[minimum size=20pt] (l4) at (12,1) {$l_a$} ;

\node (u1) at (5,0) {$a$} ;
\node (u2) at (7,0) {$a$} ;
\node (u3) at (9,0) {$b$} ;
\node (u4) at (11,0) {$b$} ;

\node[minimum size=20pt] (r0) at (12,-1) {$r_0$} ;
\node[minimum size=20pt] (r1) at (10,-1) {$r_b$} ;
\node[minimum size=20pt] (r2) at (8,-1) {$r_b$} ;
\node[minimum size=20pt] (r3) at (6,-1) {$r_b$} ;
\node[minimum size=20pt] (r4) at (4,-1) {$r_b$} ;

\node (v1) at (5,-2) {$b$} ;
\node (v2) at (7,-2) {$a$} ;
\node (v3) at (9,-2) {$b$} ;
\node (v4) at (11,-2) {$a$} ;

\draw[thick,->,>=stealth] (l0) --(l1);
\draw[thick,->,>=stealth] (l1) --(l2);
\draw[thick,->,>=stealth] (l2) --(l3);
\draw[thick,->,>=stealth] (l3) --(l4);

\draw[thick,->,>=stealth] (r0) --(r1);
\draw[thick,->,>=stealth] (r1) --(r2);
\draw[thick,->,>=stealth] (r2) --(r3);
\draw[thick,->,>=stealth] (r3) --(r4);



\draw[rotate around={45:(5,0)}] (4.8,-1.7) rectangle (5.2,1.7);
\draw[rotate around={45:(7,0)}] (6.8,-1.7) rectangle (7.2,1.7);
\draw[rotate around={45:(9,0)}] (8.8,-1.7) rectangle (9.2,1.7);
\draw[rotate around={45:(11,0)}] (10.8,-1.7) rectangle (11.2,1.7);

\draw[dotted] (5,-0.3) -- (5,-1.8);
\draw[dotted] (7,-0.3) -- (7,-1.8);
\draw[dotted] (9,-0.3) -- (9,-1.8);
\draw[dotted] (11,-0.3) -- (11,-1.8);

\node at (4.85,-0.6) {$\omega$};
\node at (6.85,-0.6) {$\omega$};
\node at (8.85,-0.6) {$\omega$};
\node at (10.85,-0.6) {$\omega$};

\node at (4,0) {input:};
\node at (4,-2) {output:};

\node at (1,-2) {Left automaton};
\node at (15,-2) {Right automaton};


\end{tikzpicture}
\end{center}


\caption{Automata of a bimachine $\bim$, and a run of $\bim$ on the word $aabb$.\label{fig:ex-bim}}
\end{figure}
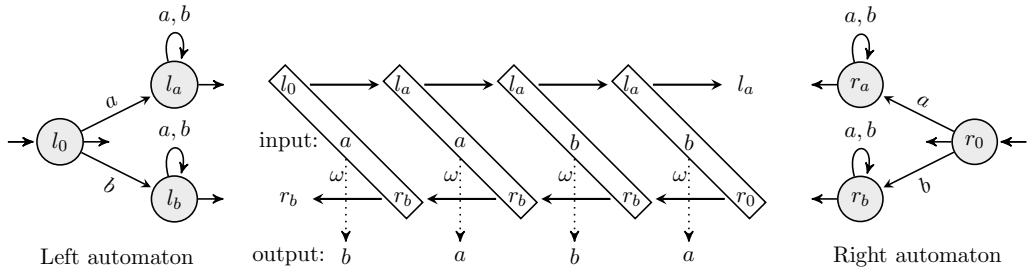

Finally, for $i=1,2$, let $\bim_i$ be a bimachine with $\laut_i$ and
$\raut_i$ as left and right automata. We say that $\bim_1$ is finer
than $\bim_2$, denoted by $\bim_1\finer \bim_2$, if $\laut_1\finer
\laut_2$ and $\raut_1\finer \raut_2$. A bimachine $\bim$ realizing a
function $f$ is \emph{minimal} if for any other bimachine $\bim'$ realizing
$f$, either $\bim'\finer \bim$ or $\bim'$ is incomparable with $\bim$.
In other words, a bimachine is minimal if there does not exist a  strictly coarser equivalent bimachine.
If $M$ is a set of bimachines, we denote by $\text{Min}(M)
\subseteq M$
the set of minimal bimachines among $M$.

\subsection{Bimachine minimization}%
\label{subsec:bim-min}
Sequential transducers can be minimized by producing the outputs as soon as possible, and this can be done in \ptime (see~\cite{Choffrut03} and Section~\ref{subsec:min-seq}).
For transducers in general, no such procedure is known.
Nevertheless, it was shown in~\cite{ReutenauerS91} that, for a bimachine, the left
automaton can be minimized when the right automaton is fixed.
This is called left minimization, and it is done in the same ``as soon as possible''
spirit but using the look-ahead information given by the right
automaton. Symmetrically, one can minimize the right automaton when
the left automaton is fixed (called right minimization).
By applying left minimization followed by right
minimization, we obtain a minimal bimachine.
We describe this minimization, generalize some properties from~\cite{ReutenauerS91}, and provide a \ptime algorithm for bimachine
minimization.

For $i=1,2$, let $\bim_i$ be a bimachine realizing $f$ with $\laut_i$ and
$\raut_i$ as left and right automata. We say that $\bim_1$ is finer
than $\bim_2$, denoted by $\bim_1\finer \bim_2$, if $\laut_1\finer
\laut_2$ and $\raut_1\finer \raut_2$.
If $\bim_1\finer\bim_2$, $\bim_2\finer\bim_1$ and $\sem{\bim_1}=\sem{\bim_2}$ then we say that $\bim_1$ and $\bim_2$ are equal up to state renaming and output shifting and we denote it by $\bim_1\bimeq\bim_2$.
A bimachine $\bim$ realizing a
function $f$ is \emph{minimal} if for any other bimachine $\bim'$ realizing
$f$, either $\bim'\finer \bim$ or $\bim'$ is incomparable with $\bim$.
In other words, a bimachine is minimal if there does not exist a  strictly coarser equivalent bimachine.
If $M$ is a set of bimachines, we denote by $\text{Min}(M)
\subseteq M$
the set of minimal bimachines among $M$.

\subsubsection{Left and right minimization}%
\label{subsubsec:lr-minimization}
Let $f$ be a transduction realized by a bimachine $\bim$ with a right automaton $\raut$.
We describe the construction of a bimachine
$\leftcan_f(\bim)=\tuple{\leftcan_f(\raut),\raut,\bout,\lfinal,\rfinal}$
which realizes $f$ and has the minimal left automaton among all
bimachines realizing $f$ with right automaton
$\raut$. We will often write $\leftcan (\bim)$ and $\leftcan(\raut)$ when the transduction is clear from context.
Intuitively, $\bim$ can be seen as a sequential transducer processing
input words the positions of which are annotated by the class (in the
transition congruence of $\raut$) of the suffix from that
position. Then, left minimization can be seen as minimizing this
sequential transducer. While it is a good intuition behind the
construction, the explicit annotation is however not necessary and we rather give a direct construction that
does not modify the input words.

We first define a family of new transductions, for any words $u,w$ let:
\[
    \widehat{f}_{\cla w_{\raut}}(u)=\bigwedge\set{f(uv)|\ v\in \cla w_{\raut}\cap u^{-1}\dom (f)}
\]
The word $\widehat{f}_{\cla w_{\raut}}(u)$ is the longest possible output from reading $u$, with the look-ahead information that the suffix is in $\cla w_\raut$.
The states of $\leftcan(\raut)$ will be the classes of the right congruence $\sim_L$ defined for any words $u,v$ by $u\sim_L v$ if:
\[
    \left\lbrack \begin{array}{ll}
    \forall w\in \Sigma^*,\ uw\in \dom(f) \Leftrightarrow vw\in \dom(f) \text{ and,}\\
    \text{if } uw\in \dom(f), \text{ then } {\widehat{f}_{\cla w_{\raut}}(u)}^{-1}f(uw)={\widehat{f}_{\cla w_{\raut}}(v)}^{-1}f(vw)
    \end{array}\right.
\]

\noindent
We define the automaton $\leftcan(\raut)=\tuple{Q,\Delta,\set{l_0},F}$ from the congruence $\sim_L$:
\begin{itemize}
\item $Q=\Sigma^*/_{\sim_L}$
\item $\Delta=\set{(\cla w_{\sim_L},\sigma,\cla{w\sigma}_{\sim_L})|\ w\in\Sigma^*,\sigma\in \Sigma}$
\item $l_0=\cla \epsilon_{\sim_L}$
\item $F=\set{\cla w_{\sim_L}|\ w\in \dom(f)}$

\end{itemize}
Now we can define the outputs of the bimachine:
\begin{itemize}
\item $\bout(\cla u_{\sim_L},\sigma,\cla w_\raut)={\widehat{f}_{\cla {\sigma w}_{\raut}}(u)}^{-1}\widehat{f}_{\cla w_{\raut}}(u\sigma)$
\item $\lfinal(\cla w_\raut)=\widehat{f}_{\cla w_{\raut}}(\epsilon)$ for $w\in \dom(f)$
\item $\rfinal(\cla u_{\sim_L})={\widehat{f}_{\cla \epsilon_{\raut}}(u)}^{-1} f(u)$ for $u\in \dom(f)$
\end{itemize}
Since we take a definition of $\sim_L$ slightly different from~\cite{ReutenauerS91}, we show, using the same idea as in~\cite{Choffrut03} for sequential transductions, that $\sim_L$ is
indeed a right congruence and that the outputs of the bimachine are
well-defined.
\begin{lem}\label{lem:well-defined}
$\sim_L$ is a right congruence and the output functions of
$\leftcan_f(\bim)$ are well-defined.
\end{lem}
\begin{proof}
Let $u\sim_L v$, $\sigma$ be a letter and let $w$ be a word in $ {(u\sigma)}^{-1}\dom(f)$.
We want to show that ${\widehat{f}_{\cla w_{\raut}}(u\sigma)}^{-1}f(u\sigma w)={\widehat{f}_{\cla w_{\raut}}(v\sigma)}^{-1}f(v\sigma w)$.
We know that ${\widehat{f}_{\cla {z}_{\raut}}(u)}^{-1}f(uz)={\widehat{f}_{\cla z_{\raut}}(v)}^{-1}f(vz)$ for $z\in u^{-1}\dom(f)$, and let us denote this word by $g(z)$.
Then $f(u\sigma w)=\widehat{f}_{\cla {\sigma w}_{\raut}}(u)g(\sigma w)$ and we have:
\[
    \begin{array}{rcl}
    {\widehat{f}_{\cla w_{\raut}}(u\sigma)}^{-1}f(u\sigma w)&=&\tuple{\bigwedge_{z\in\cla w_\raut } f(u\sigma z)}^{-1}f(u\sigma w)\\
    &=&\tuple{\bigwedge_{z\in\cla w_\raut } \widehat{f}_{\cla {\sigma w}_{\raut}}(u)g(\sigma z)}^{-1}f(u\sigma w)\\
    &=&\tuple{\widehat{f}_{\cla {\sigma w}_{\raut}}(u)\bigwedge_{ z\in\cla w_\raut } g(\sigma z)}^{-1}f(u\sigma w)\\
    &=&\tuple{\bigwedge_{ z\in\cla w_\raut } g(\sigma z)}^{-1}{\widehat{f}_{\cla {\sigma w}_{\raut}}(u)}^{-1}f(u\sigma w)\\
    &=&\tuple{\bigwedge_{ z\in\cla w_\raut } g(\sigma z)}^{-1}{\widehat{f}_{\cla {\sigma w}_{\raut}}(v)}^{-1}f(v\sigma w)\\
    &=&{\widehat{f}_{\cla w_{\raut}}(v\sigma)}^{-1}f(v\sigma w)
     \end{array}
\]
Hence $u\sigma \sim_L v\sigma$ and $\sim_L$ is indeed a right
congruence.

We only have left to show that the output functions
$\omega,\lambda,\rho$ are
well-defined \ie, they do not depend on the choices of $u$ and
$w$. First, it is rather immediate by definition of $\widehat{f}_{\cla
  {w}}$ that this function does not depend on $w$, and therefore by
definition of $\omega$ and $\lambda$, these two functions do not
depend on the choice of $w$. Now, let us show that $\omega$ does not
depend on the choice of $u$ either, \ie, ${\widehat{f}_{\cla {\sigma
    w}_{\raut}}(u)}^{-1}\widehat{f}_{\cla w_{\raut}}(u\sigma)={\widehat{f}_{\cla {\sigma w}_{\raut}}(v)}^{-1}\widehat{f}_{\cla w_{\raut}}(v\sigma)$ for $u\sim_L v$.
\[
    \begin{array}{rcl}
{\widehat{f}_{\cla {\sigma w}_{\raut}}(u)}^{-1}\widehat{f}_{\cla w_{\raut}}(u\sigma)&=& {\widehat{f}_{\cla {\sigma w}_{\raut}}(u)}^{-1} \widehat{f}_{\cla {\sigma w}_{\raut}}(u)\bigwedge_{ z\in\cla w_\raut } g(\sigma z)\\
&=& \bigwedge_{ z\in\cla w_\raut } g(\sigma z)\\
&=&{\widehat{f}_{\cla {\sigma w}_{\raut}}(v)}^{-1}\widehat{f}_{\cla w_{\raut}}(v\sigma)
\end{array}\]
Finally, by applying the definition of $\sim_L$ to $w=\epsilon$, one
immediately gets $\rfinal(\cla u_{\sim_L})=\rfinal(\cla
v_{\sim_L})$ for $u,v$ such that $u\sim_L v$ and $u,v\in\dom(f)$.
\end{proof}

Symmetrically, one can define $\rightcan_f(\laut)$ and
$\rightcan_f(\bim)$.
The correctness and effectiveness of these constructions was shown in~\cite{ReutenauerS91}:

\begin{lem}[\cite{ReutenauerS91}]\label{lem:correctness-bim}
    $\rightcan_f(\bim)$ and $\leftcan_f(\bim)$ are equivalent to
    $\bim$, and are both computable.
\end{lem}

The following proposition was shown in the case of total transductions in~\cite{ReutenauerS91}.
In our setting (\ie, when both automata recognize the domain) we are able to extend it to arbitrary transductions.
\begin{prop}%
\label{prop:l-r-min}
Let $\bim$ be a bimachine with automata $\laut$ and $\raut$.
Then $\laut\finer\leftcan(\raut)$ and $\raut\finer\rightcan(\laut)$.
\end{prop}

\begin{proof}
Let $\bim=\tuple{\laut,\raut,\bout,\lfinal,\rfinal}$ be a bimachine realizing $f$.
We will show that $\laut\finer\leftcan(\raut)$ and the proof for the right automaton is symmetrical.
Let $u\sim_\laut v$ for two words $u,v$, we want to show that $u\sim_L v$ (using the same notation as above).
Since $\laut$ recognizes $\dom(f)$, we know that for any word $w$, $uw\in \dom(f) \Leftrightarrow vw\in\dom(f)$.
We only need to show that for any word $w\in u^ {-1}\dom(f)$, ${\widehat{f}_{\cla w_{\raut}}(u)}^{-1}f(uw)={\widehat{f}_{\cla w_{\raut}}(v)}^{-1}f(vw)$.
Let $w\in u^ {-1}\dom(f)$, we have that $f(uw)=\lfinal(\cla{uw}_\raut)\bout(\cla \epsilon_\laut,u,\cla w_\raut)\bout(\cla u_\laut,w,\cla \epsilon_\raut)\rfinal(\cla {uw}_\laut)$.
We obtain $\widehat{f}_{\cla w_\raut}(u)=\lfinal(\cla{uw}_\raut)\bout(\cla \epsilon_\laut,u,\cla w_\raut)\alpha(\cla u_\laut,\cla w_\raut) $ where:
\[\alpha(\cla u_\laut,\cla w_\raut)=\bigwedge \set{\bout(\cla u_\laut,z,\cla \epsilon_\raut)\rfinal(\cla {uz}_\laut)|\ z\sim_\raut w}\]
Finally we have:
\begin{align*}
{\widehat{f}_{\cla w_{\raut}}(u)}^{-1}f(uw)&={\alpha(\cla u_\laut,\cla w_\raut)}^{-1}\bout(\cla u_\laut,w,\cla \epsilon_\raut)\rfinal(\cla {uw}_\laut)\\
&={\alpha(\cla v_\laut,\cla w_\raut)}^{-1}\bout(\cla v_\laut,w,\cla \epsilon_\raut)\rfinal(\cla {vw}_\laut)\\
&={\widehat{f}_{\cla w_{\raut}}(v)}^{-1}f(vw)
\qedhere
\end{align*}
\end{proof}
The following result is an immediate consequence of the latter proposition:
\begin{cor}\label{cor:l-r-min}
    Let $\bim = (\laut,\raut,\omega,\lambda,\rho)$ be a bimachine. Then
    \[
    \begin{array}{llllll}
    \leftcan(\bim) & \in & \text{Min}\{ \bim' =
    (\laut',\raut,\omega',\lambda',\rho')\mid \sem{\bim'} =
      \sem{\bim}\} \\
    \rightcan(\bim) & \in & \text{Min}\{ \bim' =
    (\laut,\raut',\omega',\lambda',\rho')\mid \sem{\bim'} =
      \sem{\bim}\}
    \end{array}
    \]
\end{cor}
\begin{proof}
    By definition, $\leftcan(\bim)$ has $\leftcan(\raut)$ and
    $\raut$ as left and right automata respectively. Let $\bim'$ be a
    bimachine with $\laut'$ and $\raut$ as left and right automata
    such that $\sem{\bim'} = \sem{\bim}$. By
    Proposition~\ref{prop:l-r-min}, one gets $\laut'\finer
    \leftcan(\raut)$. The result for $\rightcan(\bim)$ is shown symmetrically.
\end{proof}

\subsubsection{Bimachine minimization}

By applying left and right minimization successively, we show that we
obtain a minimal bimachine. According to
Lemma~\ref{lem:correctness-bim}, these two steps are effective, and we
call this procedure \emph{bimachine minimization}.

\begin{prop}\label{prop:minimalbim}
    Let $\bim$ be a bimachine realizing a transduction $f$. Then $\leftcan(\rightcan(\bim))$
and $\rightcan(\leftcan(\bim))$ are minimal bimachines realizing
$f$. Moreover, $\bim\finer \leftcan(\rightcan(\bim))$ and $\bim\finer
\rightcan(\leftcan(\bim))$.
\end{prop}

\begin{proof}
    We only show the result for $\rightcan(\leftcan(\bim))$, the other
    result being proved symmetrically. Let $\bim = (\laut, \raut, \omega, \lambda, \rho)$.
    Suppose that $\rightcan(\leftcan(\bim))$ is not minimal and let
    $\bim' = (\laut', \raut', \omega', \lambda',\rho')$ be a bimachine realizing $f$ such that
    $\rightcan(\leftcan(\bim))\finer \bim'$. Then,
    $\rightcan(\leftcan(\bim))$ has $\leftcan(\raut)\finer \laut'$ as left
    automaton, and $\rightcan(\leftcan(\raut))\finer \raut'$ as right
    automaton.

    By Proposition~\ref{prop:l-r-min}, we have
    $\laut \finer \leftcan(\raut)$ and $\raut\finer
    \rightcan(\leftcan(\raut))$, hence
    \[
    \laut \finer \leftcan(\raut)\finer \laut'\qquad \raut\finer
    \rightcan(\leftcan(\raut))\finer \raut'
    \]
    Since $\raut\finer \raut'$, we can restrict the right automaton of
    $\bim'$ to $\raut$
    and obtain a bimachine realizing $f$ with
    $\laut'$ as left automaton, and $\raut$ as right automaton. This
    bimachine just ignores the extra information given by $\raut$, by
    setting its output function to $(\cla u_{\laut'}, \sigma, \cla
    v_{\raut}) \mapsto \omega'(\cla u_{\laut'}, \sigma, \cla
    v_{\raut'})$. If $\laut'$ is strictly coarser than
    $\leftcan(\raut)$, this contradicts the minimality of
    $\leftcan(\raut)$ among bimachines realizing $f$ with $\raut$ as
    right automaton (Corollary~\ref{cor:l-r-min}).

    Now, suppose that $\raut'$ is strictly coarser than
    $\rightcan(\leftcan(\raut))$. We also show a contradiction, by a
    similar argument. Since $\leftcan(\raut)\finer \laut'$, we can
    restrict the left automaton of $\bim'$ to $\leftcan(\raut)$ (the
    output function ignores the finer information given by
    $\leftcan(\raut)$) and obtain a bimachine realizing $f$ with
    $\leftcan(\raut)$ as left automaton, and $\raut'$ as right
    automaton. This yields a contradiction since $\raut'$ is strictly
    finer than $\rightcan(\leftcan(\raut))$ and
    $\rightcan(\leftcan(\raut))$ is the minimal automaton among
    bimachines realizing $f$ with $\leftcan(\raut)$ as right
    automaton.
\end{proof}

\subsubsection{Minimizing bimachines in \ptime}
According to
Lemma~\ref{lem:correctness-bim} (shown in~\cite{ReutenauerS91}), we know that
$\leftcan(\rightcan(\bim))$ and $\rightcan(\leftcan(\bim))$ are
computable, for any bimachine $\bim$. We show here that it can be done
in \ptime.
\begin{thm}%
\label{thm:ptime-min}
Let $\bim$ be a bimachine. One can compute $\leftcan(\bim)$ and $\rightcan(\bim)$ in \ptime.
\end{thm}

\begin{proof}
Let $\bim=\tuple{\laut,\raut,\bout,\lfinal,\rfinal}$ be a bimachine realizing a transduction $f$, with automata $\laut=\tuple{Q_\laut,\Delta_\laut,\set{l_0},F_\laut}$ and $\raut=\tuple{Q_\raut,\Delta_\raut,\set{r_0},F_\raut}$.
The algorithm to obtain $\leftcan(\bim)$ works in two steps: (1) make the outputs earliest, (2) minimize the state space of the left automaton.

\paragraph*{Step 1:} We construct $\bim'=\tuple{\laut,\raut,\bout',\lfinal',\rfinal'}$, a bimachine with the same automata but with the earliest outputs, given the automaton $\raut$.
\[\bout'(\cla u_\laut,\sigma,\cla v_\raut)={\widehat{f}_{\cla {\sigma v}_{\raut}}(u)}^{-1}\widehat{f}_{\cla v_{\raut}}(u\sigma)\quad \lfinal'(\cla v_\raut)=\widehat{f}_{\cla v_\raut}(\epsilon)\quad \rfinal'(\cla u _\laut)={\widehat{f}_{\cla \epsilon_\raut}(u)}^ {-1}f(u)\]
These values can be computed in polynomial time, we use the same idea described in~\cite[Section~5]{Choffrut03}, which can be seen as the special case when the right automaton is trivial, and we give a sketch of the procedure.
We first remark that for any words $u,v$ $\widehat{f} _{\cla v_\raut}(u)=\lfinal(\cla{uv}_\raut) \bout(\cla \epsilon _\laut,u,\cla v_\raut) \alpha(\cla u_\laut,\cla v_\raut)$ with:

\[\alpha(\cla u_\laut,\cla v_\raut)=\bigwedge\set{w|\ \exists x\in \cla v_\raut,\ \bout(\cla u_\laut,x,\cla \epsilon_\raut)\rfinal(\cla{ux}_\raut)=w}\]
As in~\cite{Choffrut03}, in order to compute the values $\alpha(\cla u_\laut,\cla v_\raut)$, we consider the directed graph of the automaton $\laut\times\raut$, with the edges $((\cla u_\laut,\cla{\sigma v}_\raut),(\cla {u\sigma }_\laut,\cla{ v}_\raut))$ labelled by the outputs of $\bout(\cla u_\laut,\sigma, \cla v_\raut)$.
In order to account for the outputs of final functions, we add two vertices, a source $s$ pointing to the initial states, \ie, states $(\cla \epsilon _\laut,\cla v_\raut)$ for $v\in \dom(f)$, with an edge labelled by $\lfinal(\cla v_\raut)$, and a target $t$ pointing to states $(\cla u _\laut,\cla \epsilon_\raut)$ with an edge labelled by $\rfinal(\cla u_\laut)$.
The value $\alpha(\cla u_\laut,\cla v_\raut)$ is obtained as the longest common prefix of the labels of all the paths starting at $(\cla u_\laut,\cla v_\raut)$ and ending in $t$.
According to~\cite{Choffrut03} these values can be computed in polynomial time.

\paragraph*{Step 2:} Now we describe a minimization algorithm in the vein of Moore's minimization algorithm of DFAs.
The idea is to compute an equivalence relation over the state space of $\laut$, by successive refinements.
The only difference with Moore's algorithm is that the initial equivalence relation must be compatible with the output functions.
More precisely, if we identify states of $\laut$ with classes of $\sim_\laut$:
\[\begin{array}{lcl}
\cla u_\laut\sim_0 \cla v_\laut & \text{if} & \cla u_\laut\in F_\laut \Leftrightarrow \cla v_\laut\in F_\laut,\ \rfinal'(\cla u_\laut)=\rfinal'( \cla v_\laut ), \\
& & \text{and } \forall w,\sigma,\ \bout'(\cla u_\laut,\sigma,\cla w_\raut)=\bout'(\cla v_\laut,\sigma,\cla w_\raut)\\
\cla u_\laut\sim_{i+1} \cla v_\laut & \text{if} & \cla u_\laut\sim_{i} \cla v_\laut \text{ and } \forall\sigma,\ \cla {u\sigma}_\laut\sim_{i} \cla {v\sigma}_\laut
\end{array}\]

Since ${\sim_{i+1}}\finer {\sim_i}$ for any $i$, this sequence of relations converges in at most $|\laut|$ steps to a relation we denote by $\sim_*$.
Moreover, $\sim_0$ can be computed in \ptime and each $\sim_{i+1}$ can also be computed in \ptime from $\sim_{i}$.
The relation $\sim_*$ can be naturally extended to words by $u\sim_* v$ if $\cla u_\laut\sim_* \cla v_\laut$.

We want to show that the algorithm is correct meaning that ${\sim_*}={\sim_L}$, using the same notation as in~\ref{subsubsec:lr-minimization}.
We show by induction on $i\geq 0$ that ${\sim_L}\finer{\sim_i}$, and thus ${\sim_L}\finer{\sim_*}$.
Let $u\sim_L v$, we have that $u\in \dom(f) \Leftrightarrow v\in \dom(f)$, ${\widehat{f}_{\cla \epsilon_\raut}(u)}^ {-1}f(u)={\widehat{f}_{\cla \epsilon_\raut}(v)}^ {-1}f(v)$, and for all $w,\sigma$ we have  ${\widehat{f}_{\cla {\sigma w}_{\raut}}(u)}^{-1}\widehat{f}_{\cla w_{\raut}}(u\sigma)={\widehat{f}_{\cla {\sigma w}_{\raut}}(v)}^{-1}\widehat{f}_{\cla w_{\raut}}(v\sigma)$, hence $u\sim_0v$.
Now let us assume that ${\sim_L}\finer{\sim_i}$ for some integer $i$.
Let $u\sim_L v$, we have $u\sim_i v$, and for all $\sigma$, $u\sigma\sim_L v\sigma$ hence $u\sigma\sim_i v\sigma$ and $u\sim_{i+1}v$.

Conversely, let us show by induction on $i$ that if $u\sim_i v$ then for any $w\in \Sigma^i$, $uw\in\dom(f)\Leftrightarrow vw\in\dom(f)$ and ${\widehat{f}_{\cla w_{\raut}}(u)}^{-1}f(uw)={\widehat{f}_{\cla w_{\raut}}(v)}^{-1}f(vw)$ when $uw\in\dom(f)$.
If $u\sim_0v$ then $u\in \dom(f) \Leftrightarrow v\in \dom(f)$ and when $u\in \dom(f)$ then ${\widehat{f}_{\cla \epsilon_\raut}(u)}^{-1}f(u)={\widehat{f}_{\cla \epsilon_\raut}(v)}^{-1}f(v)$ by definition of $\sim_0$.
Let us assume that the proposition holds at some rank $i$, and let $u\sim_{i+1}v$.
Let $w=\sigma w'\in \Sigma^ {i+1}$. We have by assumption that $u\sigma\sim_i v\sigma$ which means that $u\sigma w'\in\dom(f)\Leftrightarrow v\sigma w'\in\dom(f)$.
Furthermore, if $u\sigma w'\in \dom(f)$ then ${\widehat{f}_{\cla {w'}_{\raut}}(u\sigma)}^{-1}f(u\sigma w')={\widehat{f}_{\cla {w'}_{\raut}}(v\sigma)}^{-1}f(v\sigma w')$ since $u\sigma\sim_i v\sigma$.
Since in particular, $u\sim_0v$ then ${\widehat{f}_{\cla {\sigma w'}_{\raut}}(u)}^{-1}\widehat{f}_{\cla {w'}_{\raut}}(u\sigma)={\widehat{f}_{\cla {\sigma w'}_{\raut}}(v)}^{-1}\widehat{f}_{\cla {w'}_{\raut}}(v\sigma)$.
Finally we obtain:
\[\begin{array}{rcl}
{\widehat{f}_{\cla {\sigma w'}_{\raut}}(u)}^{-1}f(u\sigma w')&=&\tuple{{\widehat{f}_{\cla {\sigma w'}_{\raut}}(u)}^{-1}\widehat{f}_{\cla { w'}_{\raut}}(u\sigma)}{\widehat{f}_{\cla { w'}_{\raut}}(u\sigma)}^{-1}f(u\sigma w') \\
&=&\tuple{{\widehat{f}_{\cla {\sigma w'}_{\raut}}(v)}^{-1}\widehat{f}_{\cla { w'}_{\raut}}(v\sigma)}{\widehat{f}_{\cla { w'}_{\raut}}(v\sigma)}^{-1}f(v\sigma w') \\
&=&{\widehat{f}_{\cla {\sigma w'}_{\raut}}(v)}^{-1}f(v\sigma w')
\end{array}\]
Hence $u\sim_{i+1}v$ which concludes the induction.
Since ${\sim_*}\finer{\sim_i}$ for any $i$, we have that if $u\sim_*v$ then for any $w\in \Sigma^ *$,  $uw\in\dom(f)\Leftrightarrow vw\in\dom(f)$ and ${\widehat{f}_{\cla w_{\raut}}(u)}^{-1}f(uw)={\widehat{f}_{\cla w_{\raut}}(v)}^{-1}f(vw)$ when $uw\in\dom(f)$. In particular we have shown that ${\sim_*}\finer{\sim_L}$ and finally ${\sim_*}={\sim_L}$.
\end{proof}

As a result of applying twice this procedure successively, one obtains
that bimachine minimization is in \ptime.

\begin{cor}
    Let $\bim$ be a bimachine. The minimal bimachines $\leftcan(\rightcan(\bim))$ and
    $\rightcan(\leftcan(\bim))$ can be computed in \ptime.
\end{cor}

\subsection{Canonical bimachine}\label{subsec:canbim}
We have seen how to minimize the left automaton of a bimachine with
respect to a fixed right automaton. This minimization is canonical in
the following sense: as long as two bimachines $\bim_1,\bim_2$ realize the same
transduction and have the same right automaton, the left minimization
produces the same bimachine $\leftcan_f(\bim_1) = \leftcan_f(\bim_2)$. We describe here a
canonical way of defining a right automaton $\raut_f$, which
yields, by left minimisation, a canonical bimachine denoted by $\bim_f$. This
canonical bimachine has been initially defined in~\cite{ReutenauerS91}. We recall its construction and exhibit some of
its useful properties.

\subsubsection{Canonical right automaton}
The contribution of~\cite{ReutenauerS91} relies on the existence of a canonical left congruence which yields a canonical right automaton.
This canonical right automaton can be thought of as the minimal amount of look-ahead information needed to realize the transduction.

The \emph{left congruence} of a transduction $f$ is defined by
$u\leftcong_f v$ if for any word $w$, $wu\in \dom(f) \Leftrightarrow
wv\in \dom(f)$ and $\sup\set{\dist{f(wu),f(wv)}|\ wu\in\dom(f)}<\infty$.
Intuitively, this congruence says that the two suffixes $u$
and $v$ have the same effect with respect to membership to the domain,
and that $f(wu)$ and $f(wv)$ are equal up to a suffix the length of
which is bounded by some constant depending only on $u$ and $v$.
For a rational transduction, this congruence has always finite index.
The converse does not hold, however a transduction is rational if and only its left congruence has finite index and it preserves rational languages by inverse image as it was shown in~\cite{ReutenauerS91}.
For the rest of this section, $\cla w _{\leftcong_f}$ will be denoted by $\cla w$.
The \emph{canonical right automaton} of $f$ is defined naturally from $f$: $\raut_f=\tuple{\Sigma^*/_{\leftcong_f},\Delta,\set{\cla \epsilon},F}$ with $\Delta=\set{\tuple{\cla{\sigma w},\sigma,\cla w}|\ w\in \Sigma^*, \sigma\in \Sigma}$ and $F=\set{\cla w|\ w\in \dom(f)}$.

We define symmetrically $\rightcong_f$ the \emph{right congruence} of a transduction $f$, and the associated left automaton $\laut_f$.

The automaton $\raut_f$ (resp.\ $\laut_f$) is minimal in the sense that any bimachine realizing the same function must have a finer right (resp.\ left) automaton.
\begin{prop}%
\label{prop:l-r-can}
Let $\bim=\tuple{\laut,\raut,\bout,\lfinal,\rfinal}$ be a bimachine realizing $f$.
Then $\laut\finer\laut_f$ and $\raut\finer\raut_f$.
\end{prop}

\begin{proof}
Let $\bim=\tuple{\laut,\raut,\bout,\lfinal,\rfinal}$ be a bimachine realizing $f$.
We only show that $\raut\finer\raut_f$, the proof for the left automaton being symmetrical.
Let $u\sim_\raut v$ for two words $u,v$, we want to show that $u\leftcong_f v$.
We know that $\raut$ recognizes $\dom(f)$, hence for any word $w$, $wu\in \dom(f)\Leftrightarrow wv\in \dom(f)$.
Let $w$ be a word such that $wu\in \dom(f)$.
Since $\cla u_\raut=\cla v _\raut$ we have:
\[\begin{array}{rcl}
f(wu)&=&\lfinal(\cla{wu}_\raut)\bout(\cla \epsilon_\laut,w,\cla u_\raut)\bout(\cla w_\laut,u,\cla \epsilon _\raut)\rfinal(\cla {wu}_\laut)\\
f(wv)&=&\lfinal(\cla{wu}_\raut)\bout(\cla \epsilon_\laut,w,\cla u_\raut)\bout(\cla w_\laut,v,\cla \epsilon _\raut)\rfinal(\cla {wv}_\laut)
\end{array}\]
Hence we can bound the distance between $f(wu)$ and $f(wv)$ regardless of $w$:
\[\begin{array}{rcl}
\dist{f(wu),f(wv)} &\leq & |\bout(\cla w_\laut,u,\cla \epsilon _\raut)\rfinal(\cla {wu}_\laut)|+|\bout(\cla w_\laut,v,\cla \epsilon _\raut)\rfinal(\cla {wv}_\laut)|\\
&\leq& k(|u|+|v|+2)
\end{array}\]
where $k$ is the maximum length of a word in the ranges of $\bout$ and $\rfinal$ (considering $\bout$ over single letters).
Hence we have shown that $u\leftcong_f v$ which concludes the proof.
\end{proof}

Let us show a very similar proposition which will help establish the equivalence between $\var$-transducers and $\var$-bimachines.
\begin{prop}%
\label{prop:trans-rightcan}
Let $\trans$ be a transducer, with underlying automaton $\aut$, realizing a transduction $f$.
Then ${\approx_\aut}\finer {\leftcong_f}$ and ${\approx_\aut}\finer {\rightcong_f}$.
\end{prop}

\begin{proof}
Let $\trans=\tuple{\aut,\out,\init,\final}$ be a transducer realizing $f$ with $\aut=\tuple{Q,\Delta,I,F}$ its underlying automaton.
We will show that ${\approx_\aut}\finer {\leftcong_f}$, the proof for the right congruence being symmetrical.
Let $u\approx_\aut v$, we have to show that $u\leftcong_f v$.
Since $\aut$ recognizes the domain of $f$, we have for any word $w$ that $wu\in \dom(f)\Leftrightarrow wv\in \dom(f)$.
Let $w\in u^{-1}\dom(f)$, we want to show that $\dist{f(wu),f(wv)}$ does not depend on $w$.
Let $q_I\xrightarrow{w|x}_\trans q \xrightarrow{u|y}_\trans q_F$ denote an accepting run of $\trans$ over $wu$.
Then there is an accepting run  $q_I\xrightarrow{w|x}_\trans q \xrightarrow{v|z}_\trans q_F$.
Let $i=\init(q_I)$ and $t=\final(q_F)$. Then we have:

\[\begin{array}{rcl}
\dist{f(wu),f(wv)}&=&\dist{ixyt,ixzt}\\
&=&\dist{yt,zt}\\
&\leq& k(|u|+|v|+2)
\end{array}\]
where $k$ is the maximal length of a word in the ranges of $\out$ and $\final$.
\end{proof}

\subsubsection{The canonical bimachine}
We can now define the canonical bimachine from~\cite{ReutenauerS91}.
Note that for a given transduction $f$, realized by a bimachine $\bim$, the left minimization of $\bim$ only depends on the right automaton of the bimachine.
Hence we can define the \emph{canonical bimachine} associated with $f$ by $\bim_f=\tuple{\leftcan_f(\raut_f),\raut_f,\bout_f,\lfinal_f,\rfinal_f}$ with:

\[\begin{array}{rcll}
\bout_f(\cla u _{\sim_L},\sigma, \cla w_{\raut_f})&=&{\widehat{f}_{\cla{\sigma w}_{\raut_f}}(u)}^ {-1}\widehat{f}_{\cla{ w}_{\raut_f}}(u\sigma)&\\
\lfinal_f( \cla w_{\raut_f})&=&\widehat{f}_{\cla{w}_{\raut_f}}(\epsilon)&\text{for } w\in \dom(f) \\
\rfinal_f(\cla u _{\sim_L})&=&{\widehat{f}_{\cla{\epsilon}_{\raut_f}}(u)}^{-1} f(u)& \text{for } u\in \dom(f)
\end{array}\]
This bimachine is called canonical because its definition does not depend on the description of $f$.
Furthermore, as it was shown in~\cite{ReutenauerS91}, this machine is computable from a transducer (or a bimachine) realizing $f$.
Again by symmetry, there is actually a second canonical bimachine with $\laut_f$ and $\rightcan(\laut_f)$ as its automata.
Note that the canonical bimachine is, in particular, minimal: right minimization of the canonical bimachine cannot yield a coarser right automaton than $\raut_f$ (Proposition~\ref{prop:l-r-can}).



\section{Algebraic characterization of rational transductions}%
\label{sec:algebraic_rational}

The purpose of this section is to give a characterization of
$\var$-transductions, shown to be effective when $\var$ is a decidable
class of congruences. First, we show that $\var$-transductions, which
by definition are the transductions realized by (functional but possibly ambiguous) $\var$-transducers,
also correspond to the transductions realized by the bimachines whose
left and right automata are both $\var$-automata (called
$\var$-bimachines). This correspondence is not as
straightforward as it may seem and relies on bimachine minimization.
Indeed, \emph{unambiguous} $\var$-transducers are canonically
equivalent to $\var$-bimachines, as it was stated in~\cite{ReutenauerS95}, however the equivalence was unknown in the case
of functional transducers (except for $\var=\ap$~\cite{FiliotGL16}). In other words, we (non-trivially) strengthen the known
correspondence between functional transducers and unambiguous
transducers~\cite{Eilenberg74} to $\var$-transducers, which by~\cite{ReutenauerS95} implies that (functional) $\var$-transducers and
$\var$-bimachines coincide.

The canonical bimachine defined in the previous section,
while being canonical, cannot be used to test $\var$-rationality in general:
\begin{prop}%
\label{prop:can-bim-not-C}
  There exists a congruence class $\var$ and a $\var$-transduction $f$ such that
  $\bim_f$ is not a $\var$-bimachine.
\end{prop}
\begin{proof}
Consider the transducer in Figure~\ref{i-trans} and its
associated transduction $f$.
The canonical bimachine $\bim_f$ has a trivial right automaton
(with a single state) and a left automaton which is just the underlying automaton of the minimal sequential transducer of $f$.
By Proposition~\ref{prop:not-V-seq}, $f$ is not $\id$-sequential and
thus $\bim_f$ is not an $\id$-bimachine, while $f$ is an $\id$-transduction.
\end{proof}
We will see however in Section~\ref{sec:aperiodic}
that for the special case of aperiodicity, the setting is different:
we show in Theorem~\ref{thm:can-bim-ap} that
a transduction is aperiodic iff its canonical bimachine is.

Then, to decide whether a (functional) transducer realizes a
$\var$-transduction $f$, it suffices to test whether there exists a
$\var$-bimachine realizing $f$. We show --- and it is the main result of this section ---, that any transduction has
a finite number of minimal bimachines (up to output functions and
state renaming), any of which is bounded in size
by a constant that only depends on $f$. Since congruence classes are closed under taking coarser
congruences, deciding if $f$ is a $\var$-transduction can
be reduced to deciding if one of these minimal machines is a
$\var$-bimachine. We show that the set of minimal bimachines is
computable, which gives an effective procedure to decide
$\var$-rationality of $f$ as long as $\var$ is decidable. However, we
give a less naive characterisation of $\var$-rationality which also
yields a more direct procedure for testing $\var$-rationality.
The situation is depicted in Figure~\ref{fig:rational}.

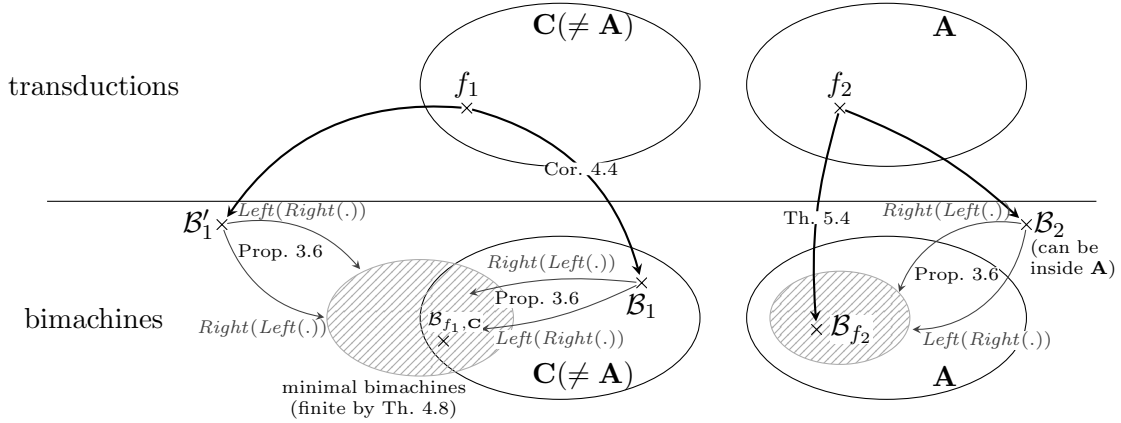
\begin{figure}


  \begin{tikzpicture}[baseline=0, inner sep=0, outer sep=0, minimum size=0pt, scale=0.31]
  \tikzstyle{cross} = [minimum size=4pt, path picture={
      \draw[black] (path picture bounding box.south east) -- (path picture bounding box.north west) (path picture bounding box.south west) -- (path picture bounding box.north east);
  }]
  \tikzstyle{projection} = [->, >=stealth, shorten >=1pt, thick, rounded corners=5]
  \tikzstyle{leftright} = [->, >=stealth, shorten >=1pt, looseness=1, color=gray!150]
  \tikzstyle{minimalzone} = [draw=gray!65, pattern=north east lines, pattern color=gray!65];

\begin{scope}

  \draw (2,12) node {transductions};
  \draw (2, 2) node {bimachines};
  \draw (0,7) -- (45,7);

  \newcommand\Vzone{(22,12) ellipse ( 6cm and 3.5cm)}
  \newcommand\Azone{(36,12) ellipse ( 6cm and 3.5cm)}
  \draw \Vzone;   \draw (23,14.5) node {$\var(\not=\ap)$};
  \draw \Azone;   \draw (38.5,14.5) node {$\ap$};

  \draw (22,2) ellipse ( 6cm and 3.5cm); \draw (23,-.5) node {$\var (\not=\ap)$};
  \draw (36,2) ellipse ( 6cm and 3.5cm); \draw (38.5,-.5) node {$\ap$};
  \draw [minimalzone] (16,2) ellipse ( 4cm and 2.5cm);
  \draw (14,-1.5) node {\tiny \begin{tabular}{c}minimal bimachines\\(finite by Th.~\ref{thm:finmin})\end{tabular}};
  \draw [minimalzone] (34,2) ellipse (3cm and 2cm);


  \draw (18,11) node [cross] (nodef1)    {}; \draw (18,12) node {$f_1$};
  \draw (25.5, 3.5) node [cross] (nodebim1) {}; \draw (25.5, 2.5) node {$\bim_1$};
  \path (nodef1) edge [projection, bend left] node [fill=white] {\tiny Cor.~\ref{cor:bim-trans}} (nodebim1);
  \draw (7.5, 6) node [cross] (nodebim1pr) {}; \draw (6.5, 6) node {$\bim_1'$};
  \path (nodef1) edge [projection, bend right] (nodebim1pr);
  \draw (13.5, 4) node (destellipse1) {}; 
  \draw (12, 2) node (destellipse2) {}; 
  \path (nodebim1pr) edge [leftright, bend left] node [above=1mm, midway] {\tiny $\leftcan(\rightcan(.))$} (destellipse1);
  \path (nodebim1pr) edge [leftright, bend right] node [below=4mm, midway] {\tiny $\rightcan(\leftcan(.))$} (destellipse2);
  \draw (10, 4.8) node {\tiny Prop.~\ref{prop:minimalbim}};
  \draw (18.5, 1.5) node (destellipse3) {}; 
  \draw (18, 3) node (destellipse4) {}; 
  \path (nodebim1) edge [leftright, bend left=10] node [below=2mm, midway, fill=white] {\tiny $\leftcan(\rightcan(.))$} (destellipse3);
  \path (nodebim1) edge [leftright, bend right=10] node [above=1mm, midway] {\tiny $\rightcan(\leftcan(.))$} (destellipse4);
  \draw (21, 2.8) node [fill=white] {\tiny Prop.~\ref{prop:minimalbim}};
  \draw (17, 1) node [cross] (nodeBf1C) {};
  \draw (17.5, 1.8) node [fill=white] {\tiny $\bim_{f_1,\var}$};


  \draw (34,11) node [cross] (nodef2) {}; \draw (34,12) node {$f_2$};
  \draw (42, 6) node [cross] (nodeB2) {}; \draw (43, 6) node {$\bim_2$};
  \draw (44, 4.5) node {\tiny \begin{tabular}{l}(can be\\inside $\ap$)\end{tabular}};
  \path (nodef2) edge [projection, bend left=10] (nodeB2);
  \draw (37, 1.5) node (destellipse5) {}; 
  \draw (36.5, 3)     node (destellipse6) {}; 
  \path (nodeB2) edge [leftright, bend left=40] node [below=4mm, midway, fill=white] {\tiny $\leftcan(\rightcan(.))$} (destellipse5);
  \path (nodeB2) edge [leftright, bend right=40] node [above=2mm, midway] {\tiny $\rightcan(\leftcan(.))$} (destellipse6);
  \draw (39, 3.8) node [fill=white] {\tiny Prop.~\ref{prop:minimalbim}};
  \draw (33, 1.5) node [cross] (nodeBf2) {};
  \draw (34.5, 1.5) node [fill=white] {$\bim_{f_2}$};
  \path (nodef2) edge [projection, bend right=10] node [fill=white] {\tiny Th.~\ref{thm:can-bim-ap}} (nodeBf2);

\end{scope}
\end{tikzpicture}


  \caption{Situation for rational transductions.\label{fig:rational}}
\end{figure}

\subsection{\texorpdfstring{$\var$}{C}-transducers and \texorpdfstring{$\var$}{C}-bimachines}
Here we show in two steps that $\var$-transducers and $\var$-bimachines realize the same transductions.
The two steps are done by showing how to construct an equivalent $\var$-transducer from a $\var$-bimachine and \emph{vice versa}.
The \emph{transition congruence} of a bimachine $\bim$ with automata $\laut$ and $\raut$ is defined as  ${\approx_\bim}={\approx_\laut\sqcap \approx_\raut}$.
Intuitively two words are equivalent with respect to $\bim$ if they are equivalent with respect to both $\laut$ and $\raut$.

\begin{prop}%
\label{prop:bim-to-trans}
Let $\bim$ be a bimachine realizing $f$, then there exists an unambiguous transducer realizing $f$ with underlying automaton $\aut$ such that ${\approx_\bim}\finer{\approx_\aut}$.
\end{prop}

\begin{proof}
Let $\bim=\tuple{\laut,\raut,\bout,\lfinal,\rfinal}$ be a bimachine with automata $\laut=\tuple{L,\Delta_\laut,\set{l_0},F_\laut}$ and $\raut=\tuple{R,\Delta_\raut,\set{r_0},F_\raut}$.
We construct a transducer $\trans=\tuple{\aut,\out,\init,\final}$ realizing the same transduction with $\aut=\tuple{Q,\Delta,I,F}$ an unambiguous automaton.

Let us define $\aut$ as the product of $\laut$ and $\raut$:
\begin{itemize}
\item $Q=L\times R$
\item $\Delta=\set{((l,r),\sigma,(l',r'))|\ (l,\sigma,l')\in \Delta_\laut \text{ and } (r,\sigma,r')\in \Delta_\raut}$
\item $I=\set{l_0}\times F_\raut$
\item $F=F_\laut\times\set{r_0}$
\end{itemize}
Since $\laut$ and $\raut$ both recognize $\dom(f)$, $\aut$ also recognizes $\dom(f)$ and for any word $w$ in $\dom(f)$, $\aut$ has a unique run over $w$.
Now we can define the outputs of $\trans$.
\begin{itemize}
\item $\out((l,r),\sigma,(l',r'))=\bout(l,\sigma,r')$
\item $\init(l_0,r)=\lfinal(r)$ for $r\in F_\raut$
\item $\final(l,r_0)=\rfinal(l)$ for $l\in F_\laut$
\end{itemize}
The transducer $\trans$ is equivalent to $\bim$ by construction, so we only have left to show that ${\approx_\bim}\finer{\approx_\aut}$.
Let $u \approx_\bim v$, then we have both $u \approx_\laut v$ and $u \approx_\raut v$ which means that $u \approx_\aut v$.
\end{proof}

\begin{prop}%
\label{prop:trans-to-bim}
Let $f$ be a transduction realized by a transducer with underlying automaton $\aut$, then there exists a bimachine $\bim$ realizing $f$ such that ${\approx_\aut}\finer{\approx_\bim}$.
\end{prop}

\begin{proof}
Let $\trans=\tuple{\aut,\out,\init,\final}$ be a transducer realizing $f$ with $\aut=\tuple{Q,\Delta,I,F}$ and let, in the following, $\cla u_\aut$ denote $\cla u _{\approx_\aut}$ for any word $u$.
We construct a bimachine $\bim=\tuple{\leftcan(\raut),\raut,\bout,\lfinal,\rfinal}$.
First we define $\raut=\tuple{R,\Delta,\set{r_0},F}$ as the right automaton canonically associated with $\approx_\aut$ seen as a left congruence:
\begin{itemize}
\item $R=\Sigma^*/_{\approx_\aut}$
\item $\Delta=\set{(\cla {\sigma u}_\aut,\sigma,\cla u_\aut)|\ u\in \Sigma^*,\sigma\in \Sigma}$
\item $r_0=\cla \epsilon_\aut$
\item $F=\set{\cla u_\aut|\ u\in\dom(f)}$
\end{itemize}
Of course we have ${\approx_\aut}\finer {\sim_\raut} $ since ${\approx_\aut}= {\sim_\raut} $, and we also have according to Proposition~\ref{prop:trans-rightcan} that ${\sim_\raut}\finer{\leftcong_f}$.
Hence we can define $\bim$ with the left minimization of $\raut$ as left automaton (see
the previous section) and we obtain a bimachine realizing $f$~\cite{ReutenauerS91}.
We only have left to show that ${\approx_\aut}\finer{\sim_L}$ (using the same notations as in~\ref{subsubsec:lr-minimization}).

Let $u,v$ be two words and let us assume that $u\approx_\aut v$.
We want to show that for all $w$, $uw\in \dom(f) \Leftrightarrow vw\in \dom(f)$ and ${\widehat{f}_{\cla w_{\raut}}(u)}^{-1}f(uw)={\widehat{f}_{\cla w_{\raut}}(v)}^{-1}f(vw)$ if $uw\in \dom(f)$.
The first condition is fulfilled, since $\aut$ recognizes the domain of $f$.
Let $w$ be a word such that $uw\in \dom(f)$.
Let $\set{p_1,\ldots,p_k}$ denote the set of states which can be reached in $\aut$ by reading $u$ (or $v$) from an initial state and which can reach a final state by reading $w$.
For $i\in\set{1,\ldots,k}$ let $I\xrightarrow{u|x_i}_\trans p_i$ and $I\xrightarrow{v|y_i}_\trans p_i$ denote the runs over $u$ and $v$, respectively, where the outputs $x_i,y_i$ are uniquely defined because $\trans$ realises a function.
For any word $z$ in $\cla w_\raut$, and any $i\in\set{1,\ldots,k}$ there is a run on $z$ from $p_i$ to a final state, since $z$ and $w$ are equivalent for $\raut$.
For $i\in\set{1,\ldots,k}$  we define $\alpha_i=\bigwedge \set{\alpha\final(q)|\ p_i\xrightarrow{z|\alpha}_\trans q,\ q\in F,\ z\in \cla w_\raut}$.
By its definition we have for any $i\in\set{1,\ldots,k}$ that $\widehat{f}_{\cla w_{\raut}}(u)=x_i\alpha_i$ and $\widehat{f}_{\cla w_{\raut}}(v)=y_i\alpha_i$.
In particular for $i=1$ we have $f(uw)=x_1\alpha_1\beta$ and $f(vw)=y_1\alpha_1\beta$ where $\alpha_1\beta=\alpha \final(q)$ for some run $p_1\xrightarrow{w|\alpha}_\trans q$ and some $q$ final.
Finally we obtain:
\begin{align*}
{\widehat{f}_{\cla w_{\raut}}(u)}^{-1}f(uw)
&={(x_1\alpha_1)}^{-1}x_1\alpha_1\beta\\
&=\beta\\
&={(y_1\alpha_1)}^{-1}y_1\alpha_1\beta\\
&={\widehat{f}_{\cla w_{\raut}}(v)}^{-1}f(vw)
\qedhere
\end{align*}
\end{proof}

A direct consequence of Propositions~\ref{prop:bim-to-trans} and~\ref{prop:trans-to-bim} is that $\var$-bimachines characterize $\var$-transductions.
\begin{cor}%
\label{cor:bim-trans}
Let $\var$ be a congruence class.
A transduction is $\var$-rational if and only if it can be realized by a $\var$-bimachine.
\end{cor}

\begin{rem}
Going from a bimachine to an unambiguous transducer is done in \ptime according to the proof of Proposition~\ref{prop:bim-to-trans}.
However, as shown in the proof of Proposition~\ref{prop:trans-to-bim}, going from a transducer to a bimachine yields, in general, an exponentially larger bimachine and is thus in \exptime.
\end{rem}

\subsection{Bounding minimal bimachines}
In this subsection we show
that any minimal bimachine realizing a transduction $f$
is bounded in size by some constant that depends only on $f$. It is
based on properties of the canonical bimachine and implies in particular that the
number of minimal bimachines is finite (up to equivalence under $\bimeq$, see Subsection~\ref{subsec:bim-min}),
a result that was left open in~\cite{ReutenauerS91}.

The following proposition intuitively shows that the more information
you put in the right automaton of a bimachine, the less information
you need in the left automaton, and symmetrically.

\begin{prop}%
\label{prop:min-inverse}
Let $\bim_1$ and $\bim_2$ be bimachines realizing a transduction $f$
with automata $\laut_1,\raut_1$ and $\laut_2,\raut_2$,
respectively. The following two implications hold true:

$\raut_1\finer \raut_2\implies
\leftcan(\raut_2)\finer\leftcan(\raut_1)$ and
$\laut_1\finer \laut_2\implies\rightcan(\laut_2)\finer\rightcan(\laut_1)$.
\end{prop}

\begin{proof}
We only show the result for right automata.
The idea of the proof is that $\raut_1$ gives more information than $\raut_2$, hence $\leftcan(\raut_2)$ must compute more information than $\leftcan(\raut_1)$ in order to realize $f$.
We want to show that there exists a bimachine realizing $f$ with automata $\leftcan(\raut_2)$ and $\raut_1$.
This will show according to Proposition~\ref{prop:l-r-min} that $\leftcan(\raut_2)\finer\leftcan(\raut_1)$.
We can assume that $\raut_1$ is accessible without loss of generality.
Identifying states of $\raut_1$ with equivalence classes of $\sim_{\raut_1}$, we have a well defined function $\pi:\Sigma^*/_{\sim_{\raut_1}}\rightarrow \Sigma^*/_{\sim_{\raut_2}}$ since $\raut_1\finer \raut_2$.
Hence if we consider $\leftcan(\bim_2)=(\leftcan(\raut_2),\raut_2,\bout_2,\lfinal_2,\rfinal_2)$, then we can define $\bim=(\leftcan(\raut_2),\raut_1,\bout,\lfinal_2,\rfinal_2\circ\pi)$ where $\bout(\cla w_{\raut_1},\sigma,l)=\bout_2(\pi(\cla w_{\raut_1}),\sigma,l)$.
By construction $\bim$ realizes $f$, which concludes the proof.
\end{proof}
We now have all the ingredients necessary to bound the size of the
minimal bimachines.
\begin{lem}%
\label{lem:bound-bim}
Let $\bim = (\laut,\raut,\omega,\lambda,\rho)$ be a bimachine
realizing a transduction $f$. If $\bim$ is minimal, then
$\leftcan(\raut_f)\finer \laut$ and $\rightcan(\laut_f)\finer \raut$.
\end{lem}

\begin{proof}
We will show the lemma for $\rightcan(\leftcan(\bim))$.
Assume that $\bim$ is minimal. Consider the bimachine
$\rightcan(\leftcan(\bim))$. According to
Proposition~\ref{prop:minimalbim}, we have
$\bim\finer \rightcan(\leftcan(\bim))$ but since $\bim$ is minimal, it
implies that the left and right automata of $\bim$ and
$\rightcan(\leftcan(\bim))$ are the same, up to renaming of their
states. The only possible difference between $\bim$ and $\rightcan(\leftcan(\bim))$ lies in
their output functions. Therefore, we can assume that
$\rightcan(\leftcan(\bim)) = (\laut,\raut,\omega',\lambda',\rho')$ for
some output functions $\omega',\lambda',\rho'$, \ie, $\laut =
\leftcan(\raut)$ and $\raut = \rightcan(\leftcan(\raut))$.

Intuitively, since $\raut_f$ contains the minimum information needed for a right automaton, $\leftcan(\raut_f)$ contains the maximum information needed for a left automaton, and any additional information should be removed by minimizing.
We first apply the left minimization and obtain $\leftcan(\raut_f)$.
According to Proposition~\ref{prop:l-r-can}, we have
$\raut\finer\raut_f$ which implies, by
Proposition~\ref{prop:min-inverse}, that
$\leftcan(\raut_f)\finer\leftcan(\raut) = \laut$.
We then use the same reasoning for the second minimization:
According to Proposition~\ref{prop:l-r-can}, we have
$\leftcan(\raut)\finer\laut_f$ and by
Proposition~\ref{prop:min-inverse}, we get
$\rightcan(\laut_f)\finer\rightcan(\leftcan(\raut)) = \raut$.
\end{proof}

\begin{thm}\label{thm:finmin}
    Let $f$ be a transduction. The set of minimal bimachines realizing
    $f$ is finite (up to $\bimeq$).
    Moreover, if $f$ is given by a transducer or a bimachine, one can
    compute a set of representatives of each class of minimal
    bimachines.
\end{thm}
\begin{proof}
    The first statement is a direct consequence of
    Lemma~\ref{lem:bound-bim}, since there are
    finitely many left and right automata
    coarser than $\leftcan(\raut_f)$ and $\rightcan(\laut_f)$.

    Now, to compute a set of representatives, we compute the set
    \[
    X = \{ \rightcan(\leftcan(\bim))\mid \bim =
    (\laut,\raut,\omega,\lambda,\rho), \sem{\bim} = f\text{ and }
    \rightcan(\laut_f)\finer \raut\finer \raut_f\}
    \]
    Before proving that this set is finite and is indeed a set of
    representatives, let us explain how to compute it. First, by~\cite{ReutenauerS91}, the
    canonical bimachine $\bim_f = (\leftcan_f(\raut_f),
    \raut_f,\omega_f,\lambda_f,\rho_f)$ (defined in Section~\ref{subsec:canbim}) is computable if $f$ is given
    by a transducer or a bimachine. By symmetry, so is the canonical
    bimachine with left and right automata $\laut_f$ and
    $\rightcan_f(\laut_f)$ respectively. Then, the computation of $X$
    is done as follows:
    \begin{enumerate}
      \item compute the canonical bimachine $\bim_f = (\leftcan_f(\raut_f),
        \raut_f,\omega_f,\lambda_f,\rho_f)$,
      \item pick a right automaton $\raut$ such that
        $\rightcan_f(\laut_f)\finer \raut\finer \raut_f$ (there are
        finitely many up to state renaming) and $\sem{\raut} =
        \dom(f)$,
      \item let $\bim = (\leftcan_f(\raut_f),
        \raut,\omega'_f,\lambda'_f,\rho_f)$ where the output functions $\omega'_f$ and
        $\lambda'_f$ are defined by
        $\omega'_f(\cla{u}_{\leftcan_f(\raut_f)}, \sigma,
        \cla{v}_{\raut}) = \omega_f(\cla{u}_{\leftcan_f(\raut_f)}, \sigma,
        \cla{v}_{\raut_f})$ and $\lambda'_f(\cla{v}_{\raut}) =
        \lambda_f(\cla{v}_{\raut_f})$.
        It is well-defined since $\raut\finer \raut_f$,

      \item compute $\bim' = \rightcan(\leftcan(\bim))$
        (Theorem~\ref{thm:ptime-min}) and add $\bim'$ to $X$,

      \item go back to step 2 as long as there is a right
        automaton $\raut$ still left to pick.
    \end{enumerate}

    \noindent It remains to prove that $X$ is finite, and
    that $X$ is a set of representatives. To show that $X$ is finite,
    it suffices to remark that given two bimachines $\bim_1$ and
    $\bim_2$ with the same right automaton $\raut$, and defining the
    same transduction, we have
    $\rightcan(\leftcan(\bim_1)) = \rightcan(\leftcan(\bim_2))$. It is
    direct by definition of the operation $\leftcan(\bim)$
    which ignores the left automaton of $\bim$ as well as its output
    functions, and only depends on its right automaton and $f$ (and
    symmetrically for the operation $\rightcan(.)$). Then, finiteness
    is due to the fact that only right automata $\raut$ such that
    $\rightcan(\laut_f)\finer \raut\finer\raut_f$ are considered, and
    there are finitely many of them.

    Finally, we show that $X$ is a set of representatives, i.e., for
    any minimal bimachine $\bim$, there exists $\bim'\in X$ such that
    $\bim\bimeq \bim'$. By Lemma~\ref{lem:bound-bim} and
    Proposition~\ref{prop:l-r-can}, if $\raut$ is the right automaton
    of $\bim$, then $\rightcan_f(\laut_f)\finer \raut\finer
    \raut_f$. Let $\bim' = \rightcan(\leftcan(\bim))$. Then clearly
    $\bim'\in X$. By Proposition~\ref{prop:minimalbim} we have
    $\bim\finer \rightcan(\leftcan(\bim)) = \bim'$, and by minimality
    of $\bim$, one obtains that $\bim\bimeq\bim'$.
\end{proof}

\subsection{Characterization of \texorpdfstring{$\var$}{C}-rationality and decision}

\subsubsection{Exhaustive search}

    If $\var$ is a decidable class of congruences, then Theorem~\ref{thm:finmin}
    implies that $\var$-rationality is decidable. Indeed, it suffices
    to compute a set of representatives of the minimal bimachines
    realizing a transduction $f$ (given for instance by a transducer),
    and then to test whether one of them is in $\var$.

\begin{lem}%
\label{lem:minimalC-bim}
Let $\var$ be a congruence class and let $f$ be a transduction.
Then $f$ is a $\var$-transduction if and only if one of its minimal bimachines is a $\var$-bimachine.
\end{lem}

\begin{proof}
If a minimal bimachine of $f$ is a $\var$-bimachine, then $f$ is a $\var$-transduction according to Corollary~\ref{cor:bim-trans}.
Conversely, let us assume that $f$ is realized by a $\var$-bimachine.
Then according to Proposition~\ref{prop:minimalbim} $\leftcan(\rightcan(\bim))$ is a minimal bimachine realizing $f$ and coarser than $\bim$.
Hence $\leftcan(\rightcan(\bim))$ is a minimal $\var$-bimachine.
\end{proof}

\begin{thm}%
\label{thm:decide}
Let $\var$ be a decidable congruence class.
Then, given a transducer, one can decide if it realizes a $\var$-transduction.
\end{thm}

\begin{proof}
According to Theorem~\ref{thm:finmin} we can compute the minimal bimachines of $f$, and from Lemma~\ref{lem:minimalC-bim} we only need to check if one of the minimal bimachines is a $\var$-bimachine.
\end{proof}

\subsubsection{Alternative characterization}

In this section we consider a slightly different characterization which only requires to compute one minimal bimachine to check $\var$-rationality.

Let $\var$ be a congruence class and let $f$ be a transduction.
We define, if it exists, $\laut_f^\var$ as the finest $\var$-automaton coarser than
$\leftcan(\raut_f)$ and finer than $\laut_f$. Let us define the right congruence
$\sim_f^\var=\bigsqcap\set{\sim_\laut\mid\ {\laut}\ \text{ is a
  }\var\text{-automaton s.t. }{{\leftcan(\raut_f)}}\finer {\laut}
}$. If ${\sim_f^\var} \finer {\sim_{\laut_f}}$ then we define $\laut_f^\var$ as the left automaton associated with $\sim_f^\var$ and recognizing $\dom(f)$.
We obtain the following characterization of $\var$-rationality:

\begin{lem}%
\label{lem:alt-characterization-C}
Let $\var$ be a congruence class and let $f$ be a transduction.
Then $f$ is a $\var$-transduction if and only if the two following conditions are satisfied:
\begin{itemize}
\item ${\sim_f^\var} \finer {\sim_{\laut_f}}$
\item $\rightcan(\laut_f^\var)$ is a $\var$-automaton
\end{itemize}
\end{lem}
\begin{proof}
Let $f$ be a transduction.
We know from Lemma~\ref{lem:minimalC-bim} that $f$ is a
$\var$-transduction if and only if it is realized by a
minimal $\var$-bimachine. Now, let us assume that $f$ is realized by a
minimal $\var$-bimachine $\bim$ with left and right automata $\laut$ and
$\raut$ respectively.
We want to show that ${\sim_f^\var} \finer {\sim_{\laut_f}}$
and $\rightcan(\laut_f^\var)$ is a $\var$-automaton. By
Lemma~\ref{lem:bound-bim} we know that $\leftcan(\raut_f)\finer
\laut$. By definition of $\sim_f^\var$, it is finer than $\sim_\laut$ since $\laut$ is a
$\var$-automaton coarser than $\leftcan(\raut_f)$. By
Proposition~\ref{prop:l-r-can}, we have $\laut\finer \laut_f$ and
hence ${\sim_f^\var} \finer {\sim_{\laut_f}}$. It remains to show that
$\rightcan(\laut_f^\var)$ is a $\var$-automaton. First, there is a
bimachine $\bim'$ with left automaton $\laut_f^\var$ which realizes $f$,
obtained by substituting the left automaton $\laut$ of $\bim$ by
$\laut_f^\var$ (recall that $\laut_f^\var\finer \laut$) and changing
the output function so that the extra information given by
$\laut_f^\var$ is just ignored. Then, by minimizing the right
automaton of $\bim'$ (whose right automaton is $\raut$), one obtains
$\raut\finer \rightcan(\laut_f^\var)$. Since $\raut$ is a
$\var$-automaton, so is $\rightcan(\laut_f^\var)$.

Conversely, suppose that ${\sim_f^\var} \finer {\sim_{\laut_f}}$ and
$\rightcan(\laut_f^\var)$ is a $\var$-automaton. We show the existence
of a $\var$-bimachine realizing $f$. Let $\bim_f =
(\laut_f,\rightcan(\laut_f),\omega_f,\lambda_f,\rho_f)$ be a
canonical bimachine associated with $f$. We can turn this bimachine
into a $\var$-bimachine. First, since $\laut_f^\var\finer \laut_f$,
one can substitute in $\bim_f$ the left automaton $\laut_f$ by
$\laut_f^\var$, change its outputs so that they ignore the extra
information given by $\laut_f^\var$, and obtain a bimachine realizing
$f$ with $\laut_f^\var$ as left automaton. By applying once the right
minimization on this new bimachine, one obtains a $\var$-bimachine realizing
$f$ with $\laut_f^\var$ and $\rightcan(\laut_f^\var)$ as left and
right automata respectively.
\end{proof}
Note that we could define $\raut_f^\var$ symmetrically and have a similar characterization.

Now we define a bimachine $\bim_{f,\var}$, which will be a $\var$-bimachine if and only if $f$ is a $\var$-transduction.
Let $\laut_{f,\var}=\laut_f\sqcap \laut_f^\var$, if $\laut_f^\var$ exists, and $\laut_{f,\var}=\laut_f$ otherwise.
Then $\bim_{f,\var}$ is the bimachine obtained with left automaton $\laut_{f,\var}$, and right automaton $\rightcan(\laut_{f,\var})$ which is well defined since by definition, $\laut_{f,\var}\finer\laut_f$.

\begin{thm}
Let $\var$ be a congruence class and let $f$ be a transduction.
Then $f$ is a $\var$-transduction if and only if the bimachine $\bim_{f,\var}$ is a $\var$-bimachine.
\end{thm}

\begin{proof}
The proof mainly relies on the previous lemma.
First, if $\bim_{f,\var}$ is a $\var$-transduction, since it realizes $f$ by definition, then $f$ is a $\var$-transduction by Corollary~\ref{cor:bim-trans}.
Conversely, if $f$ a $\var$-transduction then by Lemma~\ref{lem:alt-characterization-C}, we have both:
\begin{itemize}
\item ${\sim_f^\var} \finer {\sim_{\laut_f}}$
\item $\rightcan(\laut_f^\var)$ is a $\var$-automaton
\end{itemize}
Hence $\laut_{f,\var}= \laut_f^\var$ which is a $\var$-automaton and $\rightcan(\laut_{f,\var})=\rightcan(\laut_f^\var)$ which also is a $\var$-automaton, hence $\bim_{f,\var}$ is a $\var$-bimachine.
\end{proof}

\begin{rem}
Among all the minimal bimachines one has to test in order to check whether a transduction is $\var$-rational, $\bim_{f,\var}$ has the coarsest right automaton.
Note that when $\var=\fin$, the class of finite congruences, we have $\bim_{f,\fin}=\bim_f$.
\end{rem}



\section{Characterization of aperiodic transductions}%
\label{sec:aperiodic}

In the aperiodic case (that is, for $\var=\ap$),
we give a stronger characterization than for arbitrary congruence classes, namely that any minimal bimachine of an aperiodic transduction is aperiodic.
This gives us a \pspace algorithm for deciding aperiodicity of a transduction given by a bimachine.

\subsection{Characterization of aperiodicity}
Here we show that the canonical bimachine of an aperiodic transduction is aperiodic.
As it was shown above in Proposition~\ref{prop:l-r-can}, the canonical right automaton of an aperiodic transduction is always aperiodic.
The difficulty lies in showing that the left automaton of the canonical bimachine is also aperiodic, for an aperiodic transduction.
It relies on the decomposition of a transduction as the composition of a sequential and a \emph{right sequential} transduction, which is a characterization of rational transductions~\cite{ElgotM65}.
A right sequential transduction is simply a transduction realized by a transducer with an underlying right automaton, which can be seen as a sequential transducer but reading and writing from right to left.

Let us define $\lab_\raut$ the \emph{labelling transduction} of a right automaton $\raut$ over the alphabet $\Sigma$.
Let $Q$ be the set of states of $\raut$ and let $\Sigma_\raut=\set{\sigma_q |\ \sigma\in \Sigma ,\ q\in Q}$.
We define $\lab_\raut:\Sigma^*\rightarrow \Sigma_\raut^*$ as the transduction realized by the right-sequential transducer $\trans=\tuple{\raut,\out,\bar\epsilon,\bar\epsilon}$ where $\bar\epsilon$ is the constant function which maps any element to $\epsilon$ and $\out(p,\sigma,q)=\sigma_q$.
We also define $\ell_f$, the unique transduction such that
$f=\ell_f\circ\lab_{\raut_f}$ and
whose domain is $\{\lab_\raut(u)\ |\ u\in\dom(f)\}$.
\begin{prop}
$\ell_f$ is a sequential transduction.
\end{prop}
\begin{proof}
This proposition is a corollary of the existence of the canonical bimachine $\bim_f=\tuple{\leftcan(\raut_f),\raut_f,\out,\lfinal,\rfinal}$.
Intuitively, the right automaton does a right to left pass on the input word, annotating each letter by the current state.
Then the left automaton can produce the outputs deterministically, over this enriched alphabet.
Let $\sim_L$ denote the right congruence of $\leftcan(\raut_f)$, let $\cla u$ denote $\cla u_{\leftcong_f}$ as in~\ref{subsubsec:lr-minimization} and let ${[u]}_L$ be used for ${[u]}_{\sim_L}$.
We define $\trans_\ell=\tuple{\aut_\ell,\out_\ell,\bar\epsilon,\final_\ell}$ with $\aut_\ell=\tuple{Q_\ell,\delta_\ell,\set{q_0},F_\ell}$ a sequential transducer realizing $\ell_f$.

\begin{itemize}
\item $Q_\ell=\set{(\cla u_L,\cla v)}\uplus \set{q_0}$
\item $\delta_\ell=\set{((\cla u_L,\cla {\sigma v}),\sigma_{\cla v},(\cla {u\sigma}_L,\cla v))}\cup\set{(q_0,\sigma_{\cla v},(\cla \sigma_L,\cla v))}$
\item $F_\ell=\set{(\cla u_L,\cla \epsilon)|\ u\in\dom(f)}$ and also $q_0\in F_\ell$ if $\epsilon\in \dom(f)$
\item$ \out_\ell((\cla u_L,\cla {\sigma v}),\sigma_{\cla v},(\cla {u\sigma}_L,\cla v))=\bout(\cla u_L,\sigma,\cla v) $

and $\out_\ell(q_0,\sigma_{\cla v},(\cla \sigma_L,\cla v))=\lfinal(\cla {\sigma v})\bout(\cla \epsilon,\sigma,\cla v)$
\item $\final (\cla u_L,\cla \epsilon)=\rfinal(\cla u_L)$ and $\final(q_0)=\lfinal(\cla\epsilon)\rfinal(\cla \epsilon_L)$ if $\epsilon\in \dom(f)$
\end{itemize}
Hence $\ell_f$ is a sequential transduction.
\end{proof}

Now we have to show that $\ell_f$ is also an aperiodic transduction whenever $f$ is.
\begin{prop}
If $f$ is an aperiodic transduction, then $\ell_f$ is also aperiodic.
\end{prop}

\begin{proof}
Let $\trans$ be an aperiodic transducer realizing $f$ with an underlying automaton $\aut=\tuple{Q,\Delta,I,F}$.
We define $\trans'$, a transducer with underlying automaton $\aut'=\tuple{Q',\Delta',I',F'}$ defined as the product of $\aut$ and $\raut_f$:
\begin{itemize}
\item $Q'=Q\times (\Sigma/_{\leftcong_f})$
\item $\Delta'=\set{((p,\cla{\sigma u}),\sigma,(q,\cla u))|\ (p,\sigma,q)\in \Delta}$
\item $I'=I\times (\dom(f)/_{\leftcong_f})$
\item $F'=F\times\set{\cla \epsilon}$
\end{itemize}

\noindent
The outputs of $\trans'$ are exactly the same as the ones of $\trans$, ignoring the extra information given by $\raut_f$.
The transducer $\trans'$ can thus be seen as $\trans$ with states enriched with the look-ahead information of $\raut_f$.
The automaton $\aut'$ is clearly aperiodic since it is the product of two aperiodic automata (recall that $\raut_f$ is aperiodic from Proposition~\ref{prop:trans-rightcan}).

Now we modify the automaton $\aut'$ to obtain the automaton $\aut_\ell$ of a transducer $\trans_\ell$ realizing $\ell_f$.
For each transition $(p,\cla{\sigma u})\xrightarrow\sigma_{\aut'}(q,\cla u)$ we set $(p,\cla{\sigma u})\xrightarrow{\sigma_{\cla u}}_{\aut_\ell}(q,\cla u)$.
By construction we have $f=\ell_f\circ\lab_{\raut_f}$. We have left to show that $\aut_\ell$ is aperiodic.

Let $u\in \Sigma_{\raut_f}^*$, and let $\overline \centerdot: \Sigma_{\raut_f}^*\rightarrow \Sigma^*$ denote the natural projection morphism.
Let us assume that $(p,\cla{u_1})\xrightarrow{u^n}_{\aut_\ell}(q,\cla{u_2})$ for $n$.
By definition of $\aut_\ell$, we have $\cla{u_1}=\cla{\overline u ^{n}u_2}$ and we also have $u[|u|]=\sigma_{\cla{u_2}}$ for some letter $\sigma$, hence $\cla{\overline u u_2}=\cla{u_2}$ and $\cla{u_1}=\cla{u_2}$.
By aperiodicity of $\aut$ we have, for a large enough $n$, $p\xrightarrow{\overline u^{n+1}}_{\aut}q$.
Finally we obtain $(p,\cla{u_1})\xrightarrow{u^{n+1}}_{\aut_\ell}(q,\cla{u_1})$ and ${\aut_\ell}$ is aperiodic.
\end{proof}

\begin{cor}%
\label{cor:comp-ap}
If $f$ is aperiodic, then $\ell_f$ is an $\ap$-sequential transduction.
\end{cor}

\begin{proof}
  Using the two previous propositions, and Theorem~\ref{thm:a-seq} we show the corollary.
\end{proof}

From the decomposition of the previous Corollary, we can show that the canonical bimachine preserves aperiodicity.
\begin{thm}%
\label{thm:can-bim-ap}
A transduction is aperiodic if and only if its canonical bimachine is aperiodic.
\end{thm}

\begin{proof}
If the canonical bimachine is aperiodic then the transduction is aperiodic, by Corollary~\ref{cor:bim-trans}.
Conversely, let $f$ be an aperiodic transduction.
From Proposition~\ref{prop:trans-rightcan}, $\raut_f$ is aperiodic and according to Corollary~\ref{cor:comp-ap} there exists $\trans_\ell$ an aperiodic sequential transducer realizing $\ell$ such that $f=\ell\circ\lab_{\raut_f}$.
From $\trans_\ell=\tuple{\aut_\ell,\out_\ell,\init_\ell,\final_\ell}$, we will construct an aperiodic bimachine $\bim=\tuple{\mathcal D,\raut_f,\bout,\lfinal,\rfinal}$ realizing $f$.
This will conclude our proof since, according to Proposition~\ref{prop:l-r-min}, $\mathcal D\finer\leftcan(\raut_f)$, and the aperiodicity of $\mathcal D$ entails the one of $\leftcan(\raut_f)$.

The input alphabet of $\aut_\ell$ is $\Sigma_{\raut_f}$ so first we define $\overline \aut_\ell$ the automaton obtained by projecting all the input letters of the transitions.
The problem is that $\overline \aut_\ell$ is not deterministic and by determinizing it, some information which is needed to define the outputs of the bimachine could be lost.
The solution is to first take the product of $\overline \aut_\ell$ and $\raut_f$ and determinize that automaton and this will yield the automaton $\mathcal D$.

Let us detail this construction and show it is correct.
Let $\aut_\ell=\tuple{Q_\ell,\Delta_\ell,\set{q_0}, F_\ell}$, we define naturally $\overline\aut_\ell=\tuple{Q_\ell,\overline\Delta_\ell,\set{q_0}, F_\ell}$ with:
$\overline\Delta_\ell=\set{(p,\overline \sigma ,q)|\ (p, \sigma ,q)\in \Delta }$.
We finally define $\mathcal D$ as the deterministic automaton obtained from $\overline\aut_\ell\times\raut_f$ by subset construction.
States of $\mathcal D$ are thus subsets of $Q_\ell\times\Sigma^*/_{\leftcong_f}$.
The output function $\bout$ is defined by:

\[\bout(\set{(p_1,\cla{u_1}),\ldots,(p_n,\cla{u_n})},\sigma,\cla v)=\out_\ell(p_i,\sigma_{\cla v},q)\quad \text{s.t.}\ \cla{\sigma v}=\cla {u_i}\]
Let us show that such a state $p_i$ is unique which means that the outputs are well-defined.
We assume, for contradiction, that there exist two states $p_i,p_j$ such that $\cla {u_i}=\cla{\sigma v}=\cla {u_j}$.
Let $w$ be a word which reaches both $p_i$ and $p_j$ in $\overline\aut_\ell$, we can define a word $z$ over $\Sigma_{\raut_f}$ such that $\overline z=w$ and for any integer $k\in \set{1,\ldots,|w|}$, $z[k]=w{[k]}_c$ where $c$ is the class of the word $w[k+1:|w|]u_i$. Hence we obtain both $q_0\xrightarrow{z}_{\aut_\ell}p_i$ and $q_0\xrightarrow{z}_{\aut_\ell}p_j$ which is in contradiction with the deterministic nature of $\aut_\ell$.

The final output functions are defined by:
\[\lfinal(\cla u)=\init_\ell(q_0)\quad \rfinal(\set{(p_1,\cla{u_1}),\ldots,(pn,\cla{u_n})})=\final_\ell(p_i)\quad \text{s.t.}\ \cla{u_i}=\cla {\epsilon}\]
Similarly, the right final output function is well-defined since $\aut_\ell$ is deterministic.

It remains to show that $\overline\aut_\ell$ is aperiodic. The conclusion will follow since aperiodicity is stable by product of automata and subset construction.
Let $u$ be a word such that $p\xrightarrow{u^n}_{\overline\aut_\ell}q$ for some integer $n$.
Let $w$ be a word which reaches a final state of $\overline\aut_\ell$ from $q$.
Let us define the word $z$ over $\Sigma_{\raut_f}$ by $\overline z =u^n$ and for $i\in\set{1,\ldots,|z|}$, $z[i]=u^n{[i]}_c$ where $c$ is the class of the word $u^n[i+1:|z|]w$ in $\raut_f$.
We have by definition that $p\xrightarrow{z}_{\aut_\ell}q$.
Since $\raut_f$ is aperiodic, we have for $m$ large enough that $\cla {u^{m+1}w}=\cla {u^{m}w}$.
Then we define $y$ such that $\overline y=u$ and for $i\in\set{1,\ldots,|u|}$, $y[i]=u{[i]}_c$ where $c$ is the class of the word $u[i+1:|u|]u^{m}w$ in $\raut_f$.
If we choose $n\geq 2m$, we have $p\xrightarrow{u^m}_{\overline\aut_\ell}p'\xrightarrow{u^m}_{\overline\aut_\ell}q$.
We have by definition of $y$ that $z=y^m z'$, which means that $p\xrightarrow{y^m}_{\aut_\ell}p'\xrightarrow{z'}_{\aut_\ell}q$.
Hence, by aperiodicity of $\aut_\ell$, $p\xrightarrow{y^{m+1}}_{\aut_\ell}p'$.
Finally, we have $p\xrightarrow{y^{m+1}z'}_{\aut_\ell}q$ and since $\overline{y^{m+1}z'}=u^{n+1}$, we obtain $p\xrightarrow{u^{n+1}}_{\overline\aut_\ell}q$.
\end{proof}

\begin{cor}%
\label{cor:ap-char}
Let $\bim$ be a bimachine realizing a transduction $f$.
Then $f$ is aperiodic if and only if $\leftcan(\rightcan(\bim))$ is aperiodic if and only if $\rightcan(\leftcan(\bim))$ is aperiodic.
\end{cor}
\begin{proof}
Using both Theorem~\ref{thm:can-bim-ap} and Lemma~\ref{lem:bound-bim} we have that these bimachines are aperiodic whenever $f$ is.
Conversely, $\leftcan(\rightcan(\bim))$ and $\rightcan(\leftcan(\bim))$
both realize $f$, so $f$ is aperiodic whenever they are, by Corollary~\ref{cor:bim-trans}.
\end{proof}

\subsection{Aperiodicity is PSPACE-complete}

Here we show that deciding whether a transduction given by a bimachine is aperiodic is \pspace-complete.
The approach is very similar to~\cite{ChoH91}: first minimize the bimachine in \ptime and then check the aperiodicity of each automaton in \pspace.
The hardness is shown by reduction from the problem of deciding whether a language given by a deterministic automaton is aperiodic, which is \pspace-complete~\cite{ChoH91}.

\subsubsection{\pspace algorithm}
The fact that the problem is in \pspace is a consequence of Corollary~\ref{cor:ap-char}.
\begin{cor}%
\label{cor:pspace}
Deciding if a transduction, given by a bimachine, is aperiodic is in \pspace.
\end{cor}

\begin{proof}
Let $\bim$ be a bimachine realizing a transduction $f$.
The algorithm to decide if $f$ is aperiodic is done in two steps: first minimize the bimachine, \ie, compute $\leftcan(\rightcan(\bim))$, then check the aperiodicity of the obtained bimachine.
According to Theorem~\ref{thm:ptime-min} the first step can be done in \ptime.
According to~\cite{ChoH91} the second step can be done in \pspace for each automaton.
Finally the characterization of Corollary~\ref{cor:ap-char} ensures that the obtained bimachine is aperiodic if and only if the transduction itself is.
\end{proof}

\subsubsection{Hardness}

Deciding the aperiodicity of a language given by a minimal deterministic automaton is \pspace-hard~\cite{ChoH91}, hence we deduce that deciding aperiodicity of a transduction given by a bimachine is \emph{a fortiori} \pspace-hard.
\begin{prop}%
\label{prop:pspace-hard}
Deciding if a transduction, given by a bimachine, is aperiodic is \pspace-hard.
\end{prop}

\begin{proof}
According to~\cite{ChoH91}, deciding whether a minimal deterministic automaton recognizes an aperiodic language is \pspace-hard.
We reduce this problem to the aperiodicity problem for a transduction given by a bimachine.
Let $\aut=(Q,\delta,\set{q_0},F)$ be a minimal deterministic automaton recognizing a language $L\subseteq\Sigma^*$.
We construct, in \ptime, a bimachine $\bim$ such that $\bim$ realizes an aperiodic transduction if and only if $\aut$ is aperiodic.
First, we assume that $\aut$ is complete since an automaton can be completed in \ptime.
Since both automata of $\bim$ must recognize the same language, we cannot have $\dom(\bim)=L$ because a right automaton recognizing $L$ may need, in general, a number of states exponentially larger than $\aut$.
Instead we define $\bim$ such that for any word $w\in \Sigma^*$, $\sem\bim (w)=r$ where $r$ is the run of $\aut$ on $w$.
Note that $\dom(\bim)=\Sigma^*$ since $\aut$ is complete.

Let $\bim=\tuple{\aut',\raut_\top,\bout,\lfinal,\rfinal}$ where $\aut'=(Q,\delta,\set{q_0},Q)$ is the same as $\aut$ but all its states are final, $\raut_\top$ is the trivial complete right automaton with one state $q_\top$ which is both initial and final.
The output functions are defined by $\bout(q,\sigma,q_\top)=q$, $\lfinal(q_\top)=\epsilon$, and $\rfinal(q)=q$ for any $q\in Q,\sigma\in\Sigma$.

We claim that $\bim$ realizes an aperiodic transduction if and only if $\aut$ is aperiodic.
Let us assume that $\aut$ is aperiodic, then $\aut'$ is aperiodic as well and $\raut_\top$ is trivially aperiodic, hence $\bim$ is aperiodic and therefore realizes an aperiodic transduction.

Conversely, let us assume that $\bim$ realizes an aperiodic transduction, which we denote by $f$.
According to Theorem~\ref{thm:can-bim-ap}, we have that the canonical bimachine $\bim_f$ is aperiodic.
Let us show that automata of $\bim$ are isomorphic to those of $\bim_f$.

Since $\dom(f)=\Sigma^*$ and for all $u,v,w\in\Sigma^*$,  $\dist{f(wu),f(wv)}\leq |u|+|v|$ we have $u\leftcong_f v$.
Hence the canonical right automaton of $f$, $\raut_f$, which is the right automaton of $\bim_f$, is trivial and isomorphic to $\raut_\top$.
Now we study the right congruence $\sim_L$ associated with $\raut_f$.

Since $\leftcong_f$ is trivial, \ie, for any words $u,v$, $ u\leftcong_f v$, we will denote $\widehat{f}_{\cla u_{\leftcong_f}}$ by $\widehat{f}$.
We have $\widehat{f}(u)=\bigwedge \set{f(uv)|\ v\in \Sigma^*}=f(u)=r_u$, the run of $\aut$ on $u$, for any word $u$.
Let $u,v$ be two words, we have for any word $w$, $uw\in \dom(f)\Leftrightarrow vw \in \dom(f)$ since $\dom(f)=\Sigma^*$.
We have ${f(u)}^{-1}f(uw)={f(v)}^{-1}f(vw) \Leftrightarrow u\sim_\aut v$ which means that $u\sim_L v \Leftrightarrow u\sim_\aut v$.

We have shown that ${\sim_L}={\sim_\aut}$ which means that $\leftcan(\raut_f)$ is isomorphic to $\aut'$, therefore $\aut'$ is aperiodic and so is $\aut$.
\end{proof}

We can finally state the main theorem of the section.
\begin{thm}
Deciding if a transduction, given by a bimachine, is aperiodic is \pspace-complete.
\end{thm}

\begin{proof}
A consequence of Corollary~\ref{cor:pspace} and Proposition~\ref{prop:pspace-hard}.
\end{proof}



\section{Logical transducers}

The theory of rational languages is rich with results connecting automata, logics and algebra and some of those results have been successfully lifted to transductions:
\MSO-transducers (\MSOTs), a logical model of transducers introduced by Courcelle~\cite{CourcelleE12}, have been shown to be equivalent to deterministic two-way transducers~\cite{EngelfrietH01} as well as a deterministic one-way model of transducers with registers~\cite{AlurC10}.
More recently, an equivalence between first-order transducers and transducers with an aperiodic transition monoid has been shown~\cite{FiliotKT14,CartonD15}.

However, \MSOTs are far more expressive than one-way transducers as they can express functions that do not preserve the order of the input word, like \emph{mirror}, which returns the mirror of a word, or \emph{copy}, which copies an input word twice.
A natural restriction on \MSOTs, called \emph{order-preserving} \MSOT, has been shown to express exactly the rational functions~\cite{Bojanczyk14, Filiot15}.
In the following we will only consider order-preserving \MSOTs and simply call them \MSOTs.

Many links have been shown between logical fragments of \MSO and algebraic varieties of monoids (see \eg~\cite{Straubing94, PinS05}) and the strength of those links is that they often give equational descriptions of the corresponding language varieties, and thus give a way to decide if a language is definable in some given logical fragment.
Our goal is to provide a framework to attempt to lift some of these results from languages to transductions.
We give sufficient conditions on \F, a logical fragment of \MSO, to decide if a transduction, given by an \MSOT, can be realized by an \F-transducer.
In particular we show that it is decidable if a transduction, given by an \MSOT, can be realized by an \FOT.
\label{sec:logics}

\subsection{\MSO transductions}

\subsubsection{Words as logical structures}
A (non-empty) word $w$ over the alphabet $\Sigma$ is seen as a \emph{logical structure}\footnote{For a definition of logical structures see \eg~\cite{EbbinghausF95}} $\tilde w$ over the signature $\struct_\Sigma=\set{{(\sigma(x))}_{\sigma\in\Sigma},x< y}$.
The \emph{domain} of $\tilde w$ is the set of positions of $w$, denoted by $\dom(w)$, the unary predicate $\sigma$ is interpreted as the positions of $w$ with letter $\sigma$ and the binary predicate $<$ is interpreted as the linear order over the positions of $w$.
In order to simplify notations, we will write $w$ instead of $\tilde w$ in the following.
\subsubsection{Monadic second-order logic}
\emph{Monadic second-order formulas} (\MSO-formulas) over $\struct_\Sigma$ are defined over a countable set of first-order variables $x,y,\ldots$ and a countable set of second-order variables $X,Y,\ldots$ by the following grammar (with $\sigma\in\Sigma$):
\[\phi::=\exists X\ \phi\ |\ \exists x\ \phi\ |\ (\phi \wedge \phi)\ |\ \neg \phi \ |\ x\in X\ |\  \sigma(x)\ |\ x < y\]
The universal quantifier as well as the other Boolean connectives are defined as usual:
$ \forall X\ \phi:=\neg \exists X\ \neg \phi $, $\forall x\ \phi:=\neg \exists x\ \neg \phi $, $(\phi_1 \vee \phi_2):=\neg (\neg\phi_1 \wedge \neg\phi_2)$, and $(\phi_1 \rightarrow \phi_2):=(\neg\phi_1 \vee \phi_2)$.
We also define the formulas $\top$ and $\bot$ as being respectively always and never satisfied.
We also allow the equality predicate $x=y$.
We do not define here the semantics of \MSO nor the standard notions of \emph{free} and \emph{bound variables} but rather refer the reader to \eg~\cite{EbbinghausF95} for formal definitions.
We recall that a \emph{closed formula}, or \emph{sentence}, is a formula without free variables.
Let $\phi$ be an \MSO-sentence, we write $w\models \phi$ to denote that $w$ satisfies $\phi$.
The \emph{language defined} by an \MSO-sentence $\phi$ is the set $\sem \phi=\set{w\ |\ w\models \phi}$.
A \emph{logical fragment} \F\ of \MSO is a subset of \MSO formulas, and an \emph{\F-language} is simply a language defined by an \F-sentence.
\begin{exas}
\

\begin{itemize}
\item The fragment of \emph{first-order formulas} (\FO) is defined as the set of formulas which never use second-order variables.
\item The fragment \FOd is the set of first-order formulas with only two variables, which can both be quantified upon any number of time.
\item The existential fragment of first-order logic, \Sig{1} is the set of first-order formulas of the form $\exists x_1\ldots \exists x_n \phi$ where $\phi$ is quantifier-free. The fragment \Bsig{1} is the closure under boolean operations of \Sig{1}.

\end{itemize}

\end{exas}

\subsubsection{Pointed words}
In the model of Courcelle, originally introduced in the more general context of graph transductions~\cite{CourcelleE12}, a transduction is defined by interpreting the predicates of the output structure over several copies of the input structure.
Since we consider order-preserving word transductions, we only need to define the unary predicates of the output structure and those are defined by formulas with one free variable.
Models for such formulas are words with a marked position, called \emph{pointed words}.
Furthermore, we choose a similar but slightly different approach from formulas with a free variable and consider formulas with a constant, which can be seen as a variable which is never quantified upon.
For \MSO formulas there is no difference between having a free variable and having a constant, but for logical fragments which restrict the number of variables, this distinction will prove useful.
Pointed words have been briefly studied in~\cite{Bojanczyk15} and we will make some of the same remarks.

A \emph{pointed word} over an alphabet $\Sigma$ is a pair $(w,i)$ with $w$ a non-empty word and $i\in \dom(w)$.
A pointed word can alternatively be seen as a logical structure over the signature $\struct_\Sigma^\con=\set{\con,{(\sigma(x))}_{\sigma\in\Sigma},x<y}$ where $\con$ is a constant symbol.

A \emph{pointed \MSO-formula} is obtained from an \MSO-formula by substituting some occurrences of first order variables inside predicates by the symbol $\con$.
We denote the set of pointed \MSO-formula by \MSOc.
In particular any \MSO formula is an \MSOc formula.
Given \F\, a fragment of \MSO, we define similarly \Fc\ the set of formulas obtained from \F-formulas by substituting in predicates some occurrences of first order variables by the symbol $\con$.
Let $\psi$ be an \MSOc-sentence, we write $(w,i)\models \psi$ to denote that the pointed word $(w,i)$ satisfies $\psi$ where the constant $\con$ is interpreted as $i$.
The \emph{pointed language defined} by an \MSOc-sentence $\psi$ is the set $\sem \psi=\set{(w,i)\ |\ (w,i)\models \psi}$.

\begin{exa}%
\label{ex:pointed-language}
Let us give an example of an \MSOc-sentence over the alphabet $\Sigma \supseteq \set{\alpha,\beta,\gamma}$:
\[\gamma(\con)\wedge \tuple{\tuple{\exists x\ x<\con \wedge \alpha(x)}\vee \tuple{\exists x\ x>\con \wedge \beta(x)}}\]
This formula is in \Bsigc{1}, \FOdc, and thus also in \FOc.
\end{exa}

\subsubsection{\MSO-transducers}
An \MSO-transducer (or \MSOT) over $\Sigma$ is a tuple:
\[\trans=\tuple{K,\phi_\dom,\tuple{\psi_v}_{v\in K}}\]
where $K$ is a finite set of words, $\phi_\dom$ is an \MSO-sentence and for all $v\in K$, $\psi_v$ is an \MSOc-sentence.
The \MSOT $\trans$ defines a transduction $\sem \trans$ of domain $\sem{\phi_\dom}$.
For a non-empty word $u\in\sem{\phi_\dom}$, we have $\sem \trans(u)=v_1\ldots v_{|u|}$ such that for $i\in\dom(u)$, $(u,i)\models \psi_{v_i}$.
We remark that in general this relation is not necessarily functional: if for some word $u$ and position $i$ we have two different words $v,v'\in K$ such that $(u,i)\models \psi_v$ and $(u,i)\models \psi_{v'}$ then $u$ can have several images. Furthermore this relation may not even be well-defined on its domain: if there is a word $u$ in the domain and a position $i$ such that for any $v\in K$ we have $(u,i)\not\models \psi_v$ then the image of $u$ is not well-defined.
However these properties are decidable (reducible to satisfiability of MSO-formulas) and in this paper we always assume
that logical transductions are well-defined and functional.
In the following we only consider transductions that are (partially) defined over \emph{non-empty words}. This can be simply encoded up to adding a special symbol.
A transduction $f:\Sigma^+\rightarrow \Sigma^*$ is called \emph{\MSOT-definable} if there exists an \MSOT $\trans$ such that for any non-empty word $u\in \Sigma^+$, $f(u)=\sem\trans(u)$.

Let us state the equivalence between functional transducers and \MSO-transducers.
\begin{thm}{\cite{Bojanczyk14,Filiot15}}
A transduction is rational if and only if it is \MSOT-definable.
\end{thm}

More generally, given a logical fragment \F, an \F-transducer over $\Sigma$ is an \MSO-transducer:
\[\trans=\tuple{K,\phi_\dom,\tuple{\psi_v}_{v\in K}}\]
where $K$ is a finite set of words, $\phi_\dom$ is an \F-sentence and for all $v\in K$, $\psi_v$ is an \Fc-sentence.

\begin{rem}
Note that \FOd transductions as we define them are more expressive than transductions using only formulas with two variables, without the $\con$ symbol.
\end{rem}

\subsection{Logic-algebra equivalence}
Our goal is to obtain an equivalence between $\var$-trans\-duc\-tions and \F-transductions (as for instance $\ap$-transductions and \FO-transductions).
In this section we give sufficient conditions under which the equivalence holds.
For that purpose, we introduce an intermediate logical model of transductions, based on pairs of formulas, which is \emph{ad hoc} to bimachines.
Hence, the equivalence between this new formalism and $\var$-transductions is straightforward, and the only remaining task is to give sufficient conditions under which this formalism is equivalent with \F-transductions.
The conditions we give are, \emph{a priori}, not necessary, however we try and give conditions that are as general as possible, so that they can be applied to most known logical fragments of \MSO. 

\subsubsection{Pairs of formulas}
For a logical fragment \F, we define 2-\F\ formulas over $\Sigma$ as follows, where $\phi_1,\phi_2$ denote closed \F-formulas and $\Gamma$ denotes a subset of $\Sigma$:
\[F::=F\vee F \mid {(\phi_1,\phi_2)}_\Gamma\]
For words $u,v$ and a letter $\sigma$, we define $(u\sigma v,|u|+1) \models F$ by induction on 2-formulas:
\[\begin{array}{lcl}
(u\sigma v,|u|+1) \models {(\phi,\phi')}_\Gamma  &\mathrm{if}& u\models \phi \ \mathrm{and}\ v\models \phi' \ \mathrm{and}\ \sigma\in\Gamma\\
(u\sigma v,|u|+1)\models F_1\vee F_2 & \mathrm{if}& (u\sigma v,|u|+1)\models F_1\ \mathrm{or}\ (u\sigma v,|u|+1)\models F_2
\end{array}\]
From this we can define the $\wedge$ operator:
\[\begin{array}{rcl}
{(\phi_1,\phi'_1)}_{\Gamma_1} \wedge {(\phi_2,\phi'_2)}_{\Gamma_2} & :=& {(\phi_1\wedge \phi_2,\phi'_1\wedge \phi'_2)}_{\Gamma_1\cap\Gamma_2}\\
(F_1\vee F_2)\wedge F &:=&(F_1\wedge F)\vee(F_2\wedge F)\\
F\wedge (F_1\vee F_2) &:=&(F\wedge F_1)\vee(F\wedge F_2)
\end{array}\]
Similarly we can define the $\neg$ operator:
\[\begin{array}{rcl}
\neg {(\phi,\phi')}_\Gamma  &:=& {(\top,\top)}_{\Sigma\setminus \Gamma} \vee {(\neg\phi,\top)}_{\Sigma}\vee {(\top,\neg \phi')}_{\Sigma}\\
\neg (F_1\vee F_2) &:=& \neg F_1\wedge\neg F_2
\end{array}\]

\begin{exa}
Let us define the pointed language of Example~\ref{ex:pointed-language} by pairs of formulas:
\[\tuple{\exists x\ \alpha(x),\top}_{\set{\gamma}} \vee \tuple{\top,\exists x\ \beta(x)}_{\set{\gamma}} \]
\end{exa}

From this alternative formalism of pairs of formulas we can define \emph{logical 2-transducers} which are defined exactly as logical transducers except that  \MSOc formulas are replaced by 2-\MSO formulas.
Given a logical fragment \F, we say that 2-\F\ and \Fc\ are \emph{equivalent} if they define the same pointed languages.

\subsubsection{Bimachines and logical transducers}

The logical formalism of pairs of formulas is \emph{ad hoc} to bimachines, thus we obtain the following unsurprising result:
\begin{lem}%
\label{lem:bim-logic}
Let $\var$ be a congruence class equivalent to a fragment \F.
A transduction is definable by a $\var$-bimachine without final outputs if and only if it is definable by a 2-\F-transducer.
\end{lem}

\begin{proof}
Let $\bim=\tuple{\laut,\raut,\bout,\bar \epsilon,\bar \epsilon}$ be a $\var$-bimachine without final outputs realizing a transduction $f$.
Let us define a 2-\F-transducer $\trans=\tuple{K,\phi_\dom,\tuple{\psi_v}_{v\in K}}$ realizing $f$.
Since $\laut$ and $\raut$ are $\var$-automata, we know that for any word $w$, there exist \F-formulas $\phi_{{\cla w}_\laut}$ and $\phi_{{\cla w}_\raut}$ which respectively recognize the languages ${\cla w}_\laut$ and ${\cla w}_\raut$.
Then we can define $\phi_\dom=\bigvee_{u\in \dom(f)} \phi_{{\cla u}_\laut}$, since $\laut$ recognizes $\dom(f)$.
Let $K=\set{v\ |\ \exists u,\sigma,w \text{ s.t. } \bout(\cla u_\laut,\sigma,\cla w_\raut)=v}$,
then for $v\in K$, we define:
\[\psi_v=\bigvee_{\bout(\cla u_\laut,\sigma,\cla w_\laut)=v}\tuple{\phi_{{\cla u}_\laut},\phi_{{\cla w}_\raut}}_{\set{\sigma}}\]
By construction, $\trans$ realizes $f$.

Conversely let $\trans=\tuple{K,\phi_\dom,\tuple{\psi_v}_{v\in K}}$ be a 2-\F-transducer realizing a transduction $f$.
We define $\bim=\tuple{\laut,\raut,\bout,\bar \epsilon,\bar \epsilon}$, a $\var$-bimachine without final outputs realizing $f$.
Given an \F-sentence $\phi$, let $\laut_\phi$ and $\raut_\phi$ denote respectively a left and a right $\var$-automaton recognizing $\sem\phi$.
Let $v\in K$, and $\psi_v=\bigvee_{i=1,\ldots,n_v}{(\theta_v^i,\chi_v^i)}_{\Gamma_v^i}$, then we define:

\[\laut=\laut_{\phi_\dom}\times \prod_{ v\in K \ i=1,\ldots,n_v} \laut_{\theta_v^i} \quad \text{and} \quad  \raut=\raut_{\phi_\dom}\times\prod_{ v\in K \ i=1,\ldots,n_v} \raut_{\chi_v^i}\]
Finally, the outputs of $\bim$ are defined by $\bout({\cla u}_\laut,\sigma,{\cla w}_\raut)=v$ if there exists $v\in K$ such that $u\models \theta_v^i$, $\sigma\in \Gamma_v^i$ and $w\models \chi_v^i$ for some $i=1,\ldots,n_v$.
Otherwise, it means that $u\sigma w$ does not belong to the domain of $f$ and we can set $\bout({\cla u}_\laut,\sigma,{\cla w}_\raut)=\epsilon$.
Again, $\bim$ realizes $f$ by construction.
\end{proof}

The previous lemma does not capture the entire class of $\var$-transductions since we restrict ourselves to $\var$-bimachines without final outputs. In order to circumvent this issue, we define the notion of an \good congruence class, a class $\var$ for which the restriction over final output does not reduce the expressiveness of $\var$-bimachines.
The \emph{empty word congruence} over an alphabet $\Sigma$ is the congruence such that for any words $u,v$, we have $u\sim v$ if $u=\epsilon \Leftrightarrow v=\epsilon$. A congruence class which for any given alphabet contains the corresponding empty word congruence, is called \emph{\good}.

\begin{prop}%
\label{prop:good}
Let $\var$ be an \good congruence class.
Then any $\var$-transduction can be realized by a $\var$-bimachine without final outputs.
\end{prop}

\begin{proof}
Let $f:\Sigma^+\rightarrow \Sigma^*$ be a $\var$-transduction, and let $\bim=\tuple{\laut,\raut,\bout,\lfinal,\rfinal}$ be a $\var$-bimachine realizing $f$.
We denote by $\sim_\epsilon$ the empty word congruence over $\Sigma$, and by $\laut_\epsilon$ and $\raut_\epsilon$ the left and right automata associated with $\sim_\epsilon$, respectively.
Since ${\sim_\epsilon}\in \var$ by hypothesis, $\laut_\epsilon$ and $\raut_\epsilon$ are $\var$-automata.
We define $\bim'=\tuple{\laut\times\laut_\epsilon,\raut\times\raut_\epsilon,\bout',\overline \epsilon,\overline \epsilon}$ with
$\bout'((\cla u _\laut,\cla u_\epsilon),\sigma,(\cla v _\raut,\cla v_\epsilon))=\bout(\cla u _\laut,\sigma,\cla v _\raut)$ if $u,v\neq \epsilon$,  $\lfinal(\cla v _\raut)\bout(\cla u _\laut,\sigma,\cla v _\raut)$ if $u=\epsilon$ and $v\neq \epsilon$,  $\bout(\cla u _\laut,\sigma,\cla v _\raut)\rfinal(\cla u _\laut)$ if $u\neq\epsilon$ and $v= \epsilon$, and $\lfinal(\cla v _\raut)\bout(\cla u _\laut,\sigma,\cla v _\raut)\rfinal(\cla u _\laut)$ if $u,v= \epsilon$.
By construction $\bim'$ realizes $f$, and since ${({\approx_\bim}\sqcap{\sim_\epsilon})} \finer {\approx_{\bim'}}$, $\bim'$ is indeed a $\var$-bimachine.
\end{proof}

\begin{cor}
Let $\var$ be an \good congruence class equivalent to a fragment \F.
Then a transduction is a $\var$-transduction if and only if it is definable by a 2-\F-transducer.
\end{cor}

\begin{proof}
According to Proposition~\ref{prop:good}, we have that $\var$-bimachines without final outputs exactly characterize $\var$-transductions.
Thus, from Lemma~\ref{lem:bim-logic} we have the equivalence between $\var$-transductions and 2-\F-transductions.
\end{proof}
From the previous corollary we can immediately deduce:
\begin{cor}%
\label{cor:good-Fc2F}
Let $\var$ be an \good congruence class equivalent to a fragment \F\ such that \Fc\ and 2-\F\ define the same pointed languages.
Then a transduction is a $\var$-transduction if and only if it is an \F-transduction.
\end{cor}

\subsubsection{From \Fc\ to 2-\F and back}%
\label{sec:Fc2F}
When \F\ is a fragment equivalent to some congruence class \var, logical transductions defined with 2-\F\ formulas trivially have the same expressive power as \var-bimachines.
In this section we shall give some sufficient properties on a logical fragment \F under which \Fc-formulas and 2-\F-formulas are equivalent, \ie, define the same pointed languages.
For this we introduce a third logical formalism for pointed words:

A pointed word over an alphabet $\Sigma$ can alternatively be seen a logical structure over $\struct_{\Sigma\uplus \dot \Sigma}=\set{{(\sigma(x))}_{\sigma\in\Sigma},{(\dot\sigma(x))}_{\sigma\in\Sigma},x<y}$, \ie, a word over an extended alphabet which contains a pointed copy of each letter.
Note that the number of pointed positions is not necessarily one.
However we will see in the next example that restricting to one pointed position can be enforced by a formula $\phi_\text{pointed}$.

\begin{exa}
  The language defined by the formula in Example~\ref{ex:pointed-language} is:
  \[\set{u\alpha u'\dot\gamma v|\ u,u',v\in \Sigma^*}\cup\set{u \dot\gamma v\beta v'|\ u,v,v'\in \Sigma^*}\]

We can define this language by an \MSO-sentence over the extended alphabet $\Sigma\uplus \dot \Sigma$.
First we define the formula:
\[\phi_{\text{pointed}}:=\tuple{\exists x\ \dot\Sigma(x)} \wedge \tuple{\forall x,y\ \dot\Sigma(x)\wedge \dot\Sigma(y) \rightarrow x=y}\]
The formula $\phi_{\text{pointed}}$ specifies that exactly one position holds a pointed letter, with predicate $\dot\Sigma(x)$ being a shortcut for $\bigvee_{\sigma\in\Sigma}\dot\sigma(x)$.
Finally we obtain the formula:
\[\phi_{\text{pointed}}\wedge \exists c\ \dot\gamma(c)\wedge \tuple{\tuple{\exists x\ x<c \wedge \alpha(x)}\vee \tuple{\exists x\ x>c \wedge \beta(x)}}\]
This is just the formula of Example~\ref{ex:pointed-language} where $\con$ is replaced by the variable $c$ which is existentially quantified, and holds a pointed letter.
\end{exa}

Our goal here is to show that \Fc-formulas and 2-\F-formulas define the same pointed languages under the following assumption (1)--(3). Assumption (4) is used to show that a corresponding congruence class would be \good. 
\begin{enumerate}
\item \Fc-formulas over an alphabet $\Sigma$ and \F-formulas over the extended alphabet $\Sigma\uplus \dot \Sigma$ define the same pointed languages.
\item A language over the alphabet $\Sigma$ is definable by an \F-formula over $\Sigma$ if and only if it is definable by an \F-formula over a larger alphabet $\Sigma\cup\Gamma$.
\item \F-languages are closed under \emph{pointed concatenation}, meaning that for any two \F-languages $L_1,L_2$ over an alphabet $\Sigma$ and a fresh symbol $\sharp$, $L_1\cdot\sharp\cdot L_2$ is an \F-language over $\Sigma\uplus\set{\sharp}$.
\item $\set{\epsilon}$ is an \F-language.
\end{enumerate}
These assumptions may seem quite strong and not necessary, and they are, however they cover several well studied logical fragments of \MSO, as
shown in Section~\ref{sec:decidable-fragments}.
\begin{rem}
It seems that most of the known logical fragments that have access to the linear order and for which there is an equivalent congruence class satisfy properties (1)--(4). 
A non-example is the fragment of first-order logic with successor, \FOsuc, which is not closed under pointed concatenation.
\end{rem}

The following results establish that under assumptions (1)--(3), 
\F, \Fc and 2-\F coincide.
Figure~\ref{fig:logics} summarizes how these assumptions are used to prove it.

\begin{figure}


  \begin{tikzpicture}[baseline=0, inner sep=0, outer sep=0, minimum size=0pt]
  \begin{scope}

    \draw(0,1) node (Fc) {\Fc};
    \draw(2,1) node (F) {\F};
    \draw[white] (Fc) -- (F) node[black,midway]{$=$} node[black,midway,above=1mm] {\small (1)};

    \draw(4,1) node (2F) {2-\F};
    \draw[white] (F) -- (2F)
    node[black,midway,above=1mm] (nsub) {$\subseteq$}
    node[black,midway,below=1mm] (nsup) {$\supseteq$};
    ;

    \draw(3,2.5) node (two) {(2)};
    \draw[white] (two) -- (nsub)
    node[black,midway] {\rotatebox{270}{$\implies$}}
    node[black,midway,right=2mm] {\small Lemma~\ref{lem:Fto2F}};

    \draw(3,-0.5) node (three) {(2)-(3)};
    \draw[white] (three) -- (nsup)
    node[black,midway] {\rotatebox{90}{$\implies$}}
    node[black,midway,right=2mm] {\small Lemma~\ref{lem:2FtoF}};

  \end{scope}
\end{tikzpicture}


  \caption{Logics \F, 2-\F and \Fc coincide under assumptions (1)--(3).}\label{fig:logics} 
\end{figure}

\begin{lem}%
\label{lem:Fto2F}
Let \F\ be a fragment equivalent to some congruence class $\var$.
Under assumption (2), any \F-definable pointed language is 2-\F-definable.
\end{lem}

\begin{proof}

Let $\phi$ be an \F-formula recognizing a pointed language over the extended alphabet.
Since \F\ is equivalent to the congruence class \var, let $\aut$ be a $\var(\Sigma\uplus \dot\Sigma)$-automaton, with transition relation $\Delta$, recognizing $L(\phi)$.
Then we define

\[\psi= \bigvee_{(p,\dot \sigma,q)\in \Delta} {(\phi_I^p,\phi_q^F)}_{\set{\sigma}}\]
where $\phi_I^p,\phi_q^F$ are defined respectively as \F-formulas recognizing the language of words going from the initial states to $p$ and from $q$ to the final states. These languages are indeed definable in \F, since they are recognized by an automaton obtained from $\aut$ by changing the initial and final states (which does not modify the transition congruence).
Since the words recognized by $\aut$ only have one pointed letter, the formulas defined are \F-formulas over the regular alphabet, under assumption (2).
\end{proof}

\begin{lem}%
\label{lem:2FtoF}
Let \F\ be a fragment equivalent to some congruence class $\var$.
Under assumption (2)--(3), any 2-\F-definable pointed language is \F-definable. 
\end{lem}

\begin{proof}
Let us show that under assumption (3), a 2-\F\ formula can be transformed into a disjunction of \F-formulas recognizing pointed concatenations of \F-languages.
Let $F=\bigvee_{1\leq i\leq n} {(\phi_i,\phi_i')}_{\Gamma_i}$ be a 2-\F\ formula.
Up to decomposing ${(\phi_i,\phi_i')}_{\Gamma_i}$ into a disjunction, we can assume that $\Gamma_i=\set{\sigma_i}$.
According to assumption (3), there exists an \F-formula $\chi_i$ such that $\sem {\chi_i}=\sem{\phi_i}\dot{\sigma}_i\sem{\phi_i'}$.
Under assumption (2), these can be seen as \F-languages over $\Sigma\uplus \dot\Sigma$.
Since \F-languages are $\var$-languages, they are closed under Boolean operations.
Hence, there exists an \F-formula $\chi$ recognizing the pointed language $\bigcup_{1\leq i\leq n} \sem{\chi_i}$.
\end{proof}

\begin{prop}%
\label{prop:Fc2F}
Let \F\ be a fragment equivalent to some congruence class $\var$.
Under assumptions (1)--(3), the logics \Fc\ and 2-\F\ are equivalent. 
\end{prop}

\begin{proof}
According to Lemmas~\ref{lem:Fto2F} and~\ref{lem:2FtoF}, under assumptions (2) and (3), a pointed language is \F-definable if and only if it is 2-\F-definable.
Then, by assumption (1) we obtain the equivalence.
\end{proof}

We give sufficient conditions under which a congruence class corresponding to a logical fragment is \good.
\begin{prop}%
\label{prop:good-enough}
Let \var be a congruence class equivalent to some fragment \F satisfying properties (1) and (4).
Then \var is \good.
\end{prop}
\begin{proof}
According to (4), $\set{\epsilon}$ is an \F-language, and according to (1), $\set{\epsilon}$ is an \F-language over any alphabet $\Sigma$.
Thus for any alphabet $\Sigma$ there exists a congruence in $\var(\Sigma)$ recognizing $\set{\epsilon}$. Since $\var$ is closed under taking coarser congruences, it contains all the empty word congruences.
\end{proof}

\begin{thm}%
\label{thm:equiv}
Let $\var$ be a congruence class equivalent to some fragment \F\ satisfying properties (1)--(4). 
Then a transduction is a $\var$-transduction if and only if it is an \F-transduction.
\end{thm}
\begin{proof}
From Proposition~\ref{prop:good-enough}, we know that \var is \good. Hence from Corollary~\ref{cor:good-Fc2F} and Proposition~\ref{prop:Fc2F}, we obtain the result.
\end{proof}

\subsection{Decidable fragments}%
\label{sec:decidable-fragments}

\begin{lem}%
\label{lem:fo-fo2}
The fragments \FO, \FOd and \Bsig{1} all satisfy properties (1)--(4). 
\end{lem}

\begin{proof}
The fragments trivially satisfy (4).
Let us show that \FOd satisfies properties (1)--(3). The proof for \FO can easily be obtained with the same reasoning. 
Let us re-state the assumptions:
\begin{enumerate}
\item \FOdc-formulas over an alphabet $\Sigma$ and \FOd-formulas over the extended alphabet $\Sigma\uplus \dot \Sigma$ define the same pointed languages.
\item A language over the alphabet $\Sigma$ is definable by an \FOd-formula over $\Sigma$ if and only if it is definable by an \FOd-formula over a larger alphabet $\Sigma\cup\Gamma$.
\item \F-languages are closed under \emph{pointed concatenation}, meaning that for any two \FOd-languages $L_1,L_2$ over an alphabet $\Sigma$ and a fresh symbol $\sharp$, $L_1\cdot\sharp\cdot L_2$ is an \FOd-language over $\Sigma\uplus\set{\sharp}$.
\end{enumerate}
Let us show (1):
Let $\phi$ be an \FOd-formula over the extended alphabet $\Sigma\uplus \dot \Sigma$ recognizing a pointed language.
We syntactically replace any atomic formula $\dot\sigma(x)$ by $x=c\wedge \sigma(x)$.
Thus we obtain an \FOdc-formula recognizing the same pointed language.
Conversely let $\phi$ be an \FOdc-formula. We define as in Example~\ref{ex:pointed-language} the formula
\[
    \phi_{\text{pointed}}:=\tuple{\exists x\ \dot\Sigma(x)} \wedge \tuple{\forall x,y\ \dot\Sigma(x)\wedge \dot\Sigma(y) \rightarrow x=y}
\]
which is in \FOd.
Then we define $\phi'$ which is obtained from $\phi$ by syntactically replacing any atomic formula $\sigma(\con)$ by $\exists x\ \dot\sigma(x)$, $\sigma(x)$ by $\sigma(x)\vee \dot\sigma(x)$, $x<\con$ by $\exists y\ \dot\Sigma(y) \wedge (x<y)$ and $\con<x$ by $\exists y\ \dot\Sigma(y) \wedge (y<x)$.
Finally we obtain the \FOd-formula $\phi_{\text{pointed}}\wedge \phi'$, over the extended alphabet $\Sigma\uplus \dot \Sigma$, recognizing the same pointed language.

Let us show (2):
Let $\phi$ be an \FOd-formula over $\Sigma$ recognizing some language $L$.
Then in particular it is an \FOd-formula over $\Sigma\cup\Gamma$.
Let $\alpha=\neg \exists x\ \bigvee_{\gamma\in \Gamma\setminus\Sigma} \gamma(x)$ be a formula which specifies that no new letter appears in a word.
Then the formula $\phi\wedge \alpha$ is an \FOd-formula over $\Sigma\cup\Gamma$ which recognizes $L$.
Conversely let $\phi$ be an \FOd-formula over $\Sigma\cup\Gamma$ recognizing some language $L \subseteq \Sigma^*$.
Then one only has to syntactically change $\gamma(x)$ by $\bot$ for every predicate $\gamma \in \Gamma\setminus\Sigma$ in $\phi$ to obtain a formula $\phi'$ over $\Sigma$.
One can see that any model of $\phi$ is a model of $\phi'$ and \emph{vice versa}.

Let us show (3):
Given a formula $\phi$, we define $\phi^{<\con}$ which is the formula $\phi$ where all the quantifications are guarded and restrict variables to positions before the position $\con$. Formally, if $\phi=\exists x\ \chi(x)$, we define $\phi^{<\con} =\exists x\ (x<\con)\wedge\chi^{<\con} (x)$.
Let $\phi_1,\phi_2$ be two  \FOd-languages $L_1,L_2$ over an alphabet $\Sigma$ and let $\sharp\notin\Sigma$ be a fresh alphabet symbol.
We define $\Phi=\phi_1^{<\con}\wedge\phi_2^{>\con}\wedge \sharp(\con)\wedge \forall x\ \sharp(x)\rightarrow (x=\con)$, which recognizes the pointed language $L_1\dot\sharp L_2$.
Finally, from (1) there exists an \FOd-formula recognizing the same language.

Now we show that the fragment \Bsig{1} satisfies properties (1)--(3). 
Let us show (1):
Let $\phi$ be a \Bsig{1}-formula over the extended alphabet $\Sigma\uplus \dot \Sigma$ recognizing a pointed language.
We syntactically replace any atomic formula $\dot\sigma(x)$ by $x=c\wedge \sigma(x)$.
Thus we obtain a \Bsigc{1}-formula recognizing the same pointed language.
Conversely let $\phi$ be a \Bsigc{1}-formula. We define the formula  $\phi_{\text{pointed}}:=\tuple{\exists x\ \dot\Sigma(x)} \wedge \tuple{\forall x,y\ \dot\Sigma(x)\wedge \dot\Sigma(y) \rightarrow x=y}$ which is in \Bsig{1}.
Then we define $\phi'$ which is obtained from $\phi$ by syntactically replacing any atomic formula $\sigma(\con)$ by $\dot\sigma(c)$, $\sigma(x)$ by $\sigma(x)\vee \dot\sigma(x)$, $x<\con$ by $(x<c)$ and $\con<x$ by $ c<x$ where $c$ is a variable which does not appear in $\phi$.
Finally we obtain the \Bsig{1}-formula $\phi_{\text{pointed}}\wedge \exists c\ \dot\Sigma(c) \wedge \phi'$, over the extended alphabet $\Sigma\uplus \dot \Sigma$, recognizing the same pointed language.

Let us show (2):
Let $\phi$ be a \Bsig{1}-formula over $\Sigma$ recognizing some language $L$.
Then in particular it is a \Bsig{1}-formula over $\Sigma\cup\Gamma$.
Let $\alpha=\neg \exists x\ \bigvee_{\gamma\in \Gamma\setminus\Sigma} \gamma(x)$ be a formula which specifies that no new letter appears in a word.
Then the formula $\phi\wedge \alpha$ is a \Bsig{1}-formula over $\Sigma\cup\Gamma$ which recognizes $L$.
Conversely let $\phi$ be a \Bsig{1}-formula over $\Sigma\cup\Gamma$ recognizing some language $L \subseteq \Sigma^*$.
Then one only has to syntactically change $\gamma(x)$ by $\bot$ for every predicate $\gamma \in \Gamma\setminus\Sigma$ in $\phi$ to obtain a formula $\phi'$ over $\Sigma$.
One can see that any model of $\phi$ is a model of $\phi'$ and \emph{vice versa}.

Let us show (3):
Given a formula $\phi$, we define $\phi^{<\con}$ which is the formula $\phi$ where all the quantifications are guarded and restrict variables to positions before the position $\con$.
Formally, if $\phi=\exists x_1,\ldots,x_n\ \chi$ where $\chi$ is quantifier-free, we define $\phi^{<\con} =\exists x_1,\ldots,x_n\ (x_1<\con)\wedge\ldots\wedge (x_n<\con)\wedge \chi$.
If $\phi=Bool(\phi_1,\ldots,\phi_n)$ where the $\phi_i$s are \Sig{1}-formulas and $Bool(v_1,\ldots,v_n)$ is some propositional formula, then we define $\phi^{<\con} =Bool(\phi_1^{<\con},\ldots,\phi_n^{<\con})$.
Let $\phi_1,\phi_2$ be two  \Bsig{1}-languages $L_1,L_2$ over an alphabet $\Sigma$ and let $\sharp\notin\Sigma$ be a fresh alphabet symbol.
We define $\Phi=\phi_1^{<\con}\wedge\phi_2^{>\con}\wedge \sharp(\con)\wedge \forall x\ \sharp(x)\rightarrow (x=\con)$, which recognizes the pointed language $L_1\dot\sharp L_2$.
Finally, from (1) there exists a \Bsig{1}-formula recognizing the same language.
\end{proof}

\begin{thm}%
\label{thm:decidable-logic}
Given a transducer realizing a transduction $f$, one can decide if:
\begin{itemize}
\item $f$ is \FO-definable.
\item $f$ is \FOd-definable.
\item $f$ is \Bsig{1}-definable.
\end{itemize}
\end{thm}

\begin{proof}
It was shown by~\cite{Schutzenberger61,McnaughtonP71} that \FO\ and $\ap$ are equivalent.
It was also shown in~\cite{TherienW98} that \FOd\ and $\da$ are equivalent.
The equivalence between \Bsig{1} and the class $\jtrivial$ of $\mathcal J$-trivial congruences is due partly to~\cite{Simon75} and can be found in~\cite{DiekertGK08}.
From Theorem~\ref{thm:equiv} and Lemma~\ref{lem:fo-fo2} we have the equivalences between \FO-transductions and $\ap$-transductions, \FOd-transductions and $\da$-transductions, and \Bsig{1}-transductions and $\jtrivial$-transductions respectively.
In the articles mentioned above it is also shown that the equivalences are effective, meaning that the congruence classes $\jtrivial$, $\da$ and $\ap$ are decidable, thus from Theorem~\ref{thm:decide} we obtain the result.
\end{proof}

\begin{rem}
Extensions of the above fragments with additional predicates can easily be shown to satisfy assumptions (1)--(4). 
A table from~\cite{DartoisP15} sums up many results of decidability concerning fragments of \FO with additional predicates.
For all these fragments, our Theorem~\ref{thm:decidable-logic} should carry over. Similarly, the proof for \Bsig{1} should easily transfer to \Bsig{i}, for any positive integer $i$.

\end{rem}


\section*{Conclusion}%
\label{sec:conclusion}
We have shown that the problem of deciding whether a bimachine realizes an aperiodic transduction is \pspace-complete. A question which remains open is whether the problem is still in \pspace if the input is given as a transducer instead of a bimachine.
Our main result is that any rational transduction can be completely characterized by a finite set of minimal bimachines, and we have introduced a logical formalism which allows to lift some known equivalences from rational languages to rational transductions.

One possible direction of investigation would be to formalize in algebraic terms sufficient, and possibly necessary, conditions under which a logic-algebra equivalence for languages can be lifted to transductions.
It would also be interesting to investigate the cases which are precisely not covered by our approach such as the logics \MSOeq or \FOsuc.
Another direction would be to try to lift the present results further to regular functions, \ie, functions realized by two-way transducers.
Indeed a characterization of a regular transduction in terms of minimal congruences would for instance yield a way to decide if a regular function is first-order definable. In~\cite{Bojanczyk14} such a characterization is given, but with the stronger \emph{origin semantics}, where every output position originates from an input position, which means that two transducers which realize the same function but produce outputs at different positions will not be considered equivalent under this semantics. Obtaining a similar characterization in origin-free semantics is a very ambitious objective and for now there is no known canonical way to define a regular function.

\section*{Acknowledgement}

The authors thank Anca Muscholl for her contribution in
early stages of this work.

\bibliographystyle{alpha}

\bibliography{main}

\newpage
\appendix

\end{document}